\documentclass[aps,nofootinbib,superscriptaddress,tightenlines,11pt]{revtex4}
\usepackage[utf8]{inputenc}
\usepackage{graphicx, float, amsmath, physics}
\usepackage{amsmath, amssymb, amsfonts, latexsym, amsthm, bbm, braket}
\usepackage[normalem]{ulem}

\usepackage{enumerate,enumitem} 
\usepackage{dsfont} 
\usepackage{tcolorbox}

\usepackage{verbatim}

\usepackage{ragged2e}
\usepackage[a4paper, total={6in, 8in}]{geometry}
\usepackage{tikz}
\graphicspath{ {./images/} }

\usepackage{hyperref}
\hypersetup{
    colorlinks = true,
    allcolors=blue}
\usepackage{cleveref}

\def\A{\mathcal{A}}
\def\B{\mathcal{B}}
\def\E{\mathcal{E}}

\def\K{\mathcal{K}}
\def\N{\mathcal{N}}

\def\a{\mathbf a}
\def\k{\mathbf k}
\def\t{\mathbf t}
\def\s{\mathbf s}
\def\z{\mathbf z}

\def\r{\mathsf r}

\def\e{\mathsf e}
\def\c{\mathsf c}
\def\q{\mathsf q}
\def\a{\mathbf a}
\def\b{\mathbf b}
\def\x{\mathbf x}
\def\y{\mathbf y}
\def\lo{\mathsf{rest}}
\def\z{\mathbf z}
\def\br{\mathbf r}
\def\bv{\mathbf v}

\def\ZO{\{0,1\}}

\newcommand{\unity}{\ensuremath{\mathbbm 1}}
\def\Span{\mathrm{span}\hspace{0.05cm}}
\newcommand{\proj}[1] {|#1 \rangle\!\langle #1|}
\DeclareMathOperator*{\EE}{\mathbb{E}}

\theoremstyle{definition}
\newtheorem*{theorem}{Theorem}
\newtheorem{lemma}{Lemma}

\begin{document}

\title{
High-rate and computationally-efficient seedless extractors for device-independent quantum cryptography
}

\author{Simone Lin}
\email{simone.lin.23@ucl.ac.uk}
\affiliation{Department of Computer Science, University College London, United Kingdom}

\author{Cameron Foreman}
\email{cameron.foreman@quantinuum.com}
\affiliation{Department of Computer Science, University College London, United Kingdom}
\affiliation{Quantinuum, Partnership House, Carlisle Place, London SW1P 1BX, United Kingdom}

\author{Llu\'is Masanes}
\email{l.masanes@ucl.ac.uk}
\affiliation{Department of Computer Science, University College London, United Kingdom}
\affiliation{London Centre for Nanotechnology, University College London, United Kingdom}

\begin{abstract}
  Device-independent (DI) quantum cryptography provides secure cryptography with minimal trust in, or characterisation of, the used quantum devices. An essential component of DI protocols is the use of randomness extractors for privacy amplification, but these typically require an initial seed of randomness that introduces a potential vulnerability. To solve this problem, the security of seedless extractors was proven in \href{https://doi.org/10.22331/q-2025-03-06-1654}{Quantum 9, 1654 (2025)}. The core idea was to use the Bell violation of the raw data, rather than its min-entropy, as the extractor promise. However, the large fluctuations in the Bell inequality used required many rounds to precisely estimate the Bell violation, consuming substantial randomness and making the protocol very inefficient. In this work, we present a new proof technique based on a truncation method that allows the user to estimate the protocol parameters with an asymptotically vanishing fraction of rounds and, as a consequence, achieves the optimal rate of one key bit per singlet. Notably, we prove this result using seedless extractors that can be implemented efficiently.
\end{abstract}
\maketitle

\section{Introduction}
Unlike traditional quantum cryptography, device-independent (DI) quantum cryptography provides security with minimal trust in, or characterisation of, the quantum hardware used to implement the protocol. This is achieved through security proofs that exploit quantum non-locality---observing correlations that violate a Bell inequality \cite{Bell1964}---rather than relying on models of the devices' internal workings.
By removing assumptions about the devices themselves, DI protocols remain secure even in the presence of hardware noise and performance fluctuations.
As a result, DI protocols have been developed for a wide range of cryptographic tasks, including secure key distribution \cite{Pironio_2009, Masanes2011,Pironio_2013,Vazirani_2014}, certified randomness generation \cite{pironio2010random,Colbeck_2011,Ac_n_2016}, and randomness amplification \cite{Colbeck_2012,brandao2016realistic,kessler2020device,foreman2023practical}.

A key step in many DI protocols is randomness extraction (e.g., for privacy amplification), in which the measurement outcomes are classically processed by a randomness extractor to produce uniformly distributed bits that are uncorrelated with any potential adversary.
To operate, randomness extractors require a statistical promise on their input; in DI protocols, this is traditionally a bound on the min-entropy of the measurement outcomes. When min-entropy is used as the extractor's promise, the extraction process necessarily requires a random \textit{seed} that is sufficiently unpredictable and independent \cite{santha1986generating}. In reality, the required resource for extraction can be even more demanding, e.g., a seed that is perfectly uniform and fully independent \cite{Ma_2013, Foreman_2025}.

However, seedless extraction becomes possible in DI quantum cryptography when the extractor relies on a statistical promise other than min-entropy. The first proof of deterministic extraction in the DI setting was given in \cite{Masanes_2009}. Subsequent work used semidefinite programming to show that the XOR of two bits, as well as several subsets of bits, is a deterministic extractor for outputs generated by independent, memoryless devices \cite{Hanggi_2009, Hanggi_2010, Hanggi_2010_thesis}. 
In our previous work \cite{seedless}, we showed that deterministic extraction is possible for DI protocols with memoryless measurement devices by constructing extractors whose security guarantee depends directly on the observed Bell-inequality violation, giving a full protocol and rate analysis. However, several important limitations remained. First, it used seedless extractors whose existence was proven theoretically, but for which efficient implementations (i.e., polynomial time computable) were not known. Second, the protocol was highly inefficient, generating fewer random bits than it consumed to estimate the Bell-inequality violation, severely limiting its practical utility.

In this work, we show that linear functions (i.e., functions represented by binary matrices) can be used as seedless extractors and, using our new proof technique, we find specific linear functions that achieve the optimal rate of one bit per singlet for the task of one-party DI randomness generation, provided the number of experimental rounds is sufficiently large. The key technical ingredient behind this advancement is a new proof technique based on a truncation method.  
The intuition is as follows: in proving security, we aim to bound the distinguishability (measured by the trace norm) between the state produced by the protocol and an ideal uniformly random state. In our earlier work, this required bounding the expectation of an $n$-fold product of Bell inequalities, which was highly sensitive to small perturbations in the global state. As a result, a large number of rounds had to be consumed to ensure that this expectation was below a desired threshold with high probability. In the present work, we overcome this challenge by truncating the initial state into a subspace giving large probability for the obtained estimation outcome, so that the expectation is directly bounded. The main difficulty is then in defining an appropriate truncation, since the eigen-basis of the estimation measurements differ from those generating the raw key. We present our results in the context of one-party DI random number generation using the CHSH inequality so that they can be readily adapted to other DI protocols and tasks. 

The remainder of this manuscript is organised as follows. In Section~\ref{sec:main_results}, we present our main results by giving a description of the protocol, introducing the setup and notation, stating our main theorem, and including a plot of the protocol rates. In Section~\ref{sec:proof}, we provide a detailed proof of the main theorem using our truncation method. This proof relies on several lemmas, which are themselves proved in the Appendix.

\section{Main result} \label{sec:main_results}

\subsection{Protocol for one-party random number generation}
\label{protocol_description}
The protocol involves two non-communicating and memoryless quantum devices, Alice and Bob, and one user who interacts with them and performs the following steps.
\begin{enumerate}[leftmargin=*]
  \item \textbf{Set parameters.} Choose:
  \begin{itemize}[leftmargin=7mm]
    \item Number of rounds, $n \in \mathbb{Z}^+$, which fixes the estimation probability $p_\e = n^{-1/3}$.
    \item Threshold winning frequency for not aborting, $w_\e \in (w_\c + \Delta, w_\q]$, where $w_\c = 3/4$ and $w_\q = \cos^2(\pi/8) = (2+\sqrt{2})/4$ are the maximum winning probabilities of the CHSH game for classical and quantum strategies, respectively.
    \item Error tolerance for the final key, $\epsilon \in (0,1]$, defined in \eqref{eq:sec cond}, which fixes the statistical uncertainty buffer $\Delta =n^{-1/3} \sqrt{2\ln(8 /\epsilon)}$.
  \end{itemize}

  \item \textbf{Initialisation.} With probability $p_\e$, independently decide for each $i \in \{1, \ldots, n\}$ whether round $i$ is used for estimation or raw-key generation. Then:
    \begin{itemize}[leftmargin=7mm]
        \item Record the set of estimation rounds $\N_\e \subseteq \{1, \ldots, n\}$ and its size $n_\e$.
        \item Record the set of generation rounds $\N_\r = \{1, \ldots, n\} \backslash \N_\e$ and its size $n_\r = n - n_{\e}$.
        \item Compute the final key length $m(w_\e, \epsilon, n)$ using formula~\eqref{main result}.
        \item Choose a binary $m\times n_\r$ extractor matrix $G$ satisfying condition~\eqref{extract condit}.
    \end{itemize}

    \item \textbf{Round execution.} In each round $i = 1, \dots, n$: 
    \begin{itemize}[leftmargin=7mm]
        \item If $i\in \N_\r$: 
        \begin{itemize}[leftmargin=7mm]
            \item Send Alice and Bob the inputs $(x_i, y_i) = (0, 0)$. 
            \item Alice returns $a_i \in \{0,1\}$, which will be part of the raw key.
        \end{itemize}
        \item If $i\in \N_\e$: 
        \begin{itemize}[leftmargin=7mm]

            \item Send Alice and Bob uniformly random inputs $(x_i, y_i) \in \{0,1\}^2$.
            \item Alice and Bob return $(a_i, b_i) \in \{0,1\}^2$. 
            \item Record the variable $z_i= a_i+b_i+x_iy_i +1 \bmod 2$, which indicates whether the CHSH game is won ($z_i=1$) or lost ($z_i=0$).
        \end{itemize}
    \end{itemize}

    \item \textbf{Verification and extraction.} After $n$ rounds, the raw key is $\mathbf{a}_\r = (a_i)_{i \in \mathcal{N}_\r} \in \{0,1\}^{n_{\r}}$, and the estimation data is $\z_\e = (z_i)_{i \in \mathcal{N}_\e} \in\ZO^{n_\e}$. 
    \begin{itemize}[leftmargin=7mm]

        \item If the number of winning estimation rounds, $|\z_\e|=\sum_{i \in \mathcal{N}_\e} z_i$, meets or exceeds the chosen threshold $|\z_\e|\geq n_\e w_\e$ then compute the final key ($\mathbf{k} = G\mathbf{a}_\r$) and end the protocol. 
        \item If $|\z_\e| < n_\e w_\e$ then abort the protocol.
    \end{itemize}
\end{enumerate}

\subsection{Setup and notation}

The protocol consists of $n$ rounds labelled by $i\in\{1, \ldots, n\}$. In round $i$, Alice and Bob use a pair of quantum systems with Hilbert spaces $\A_i$ and $\B_i$. System $\A_i$ is measured with one of two observables indexed by  $x_i \in \{0,1\}$, described by POVM elements $A_i(a_i|x_i)$, each producing the outcome $a_i\in \ZO$. Analogously, $\B_i$ is measured with the two observables described by $B_i(b_i|y_i)$, for $b_i, y_i\in \ZO$. The total Hilbert spaces of Alice, Bob and the adversary, Eve, are $\A= \bigotimes_i \A_i$, $\B= \bigotimes_i \B_i$ and $\E$, and their initial joint state is $\rho_{\A\B\E}$.
The assumption that the measurement devices have no internal memory is formalised by using separate Hilbert spaces $\A_i \B_i$ and measurements $A_i(a|x), B_i(b|y)$ in every round $i$.

In each round $i$, an independent random variable is generated with probabilities $p_\e$ and $p_\r=1-p_\e$, determining whether the round $i$ is used for estimation ($\e$) or for raw-key generation ($\r$). 
The rounds used for estimation are recorded in the set $\N_\e \subseteq \{1, \ldots, n\}$ which has cardinality $n_\e$, and the rounds used for raw key in  $\N_\r = \{1, \ldots, n\} \backslash \N_\e$ with cardinality $n_\r =n-n_\e$. 
The Hilbert spaces of the estimation systems are $\A_\e = \bigotimes_{i\in\N_\e} \A_i$ and $\B_\e = \bigotimes_{i\in\N_\e} \B_i$, while those of raw-key systems are $\A_\r = \bigotimes_{i\in\N_\r} \A_i$ and $\B_\r = \bigotimes_{i\in\N_\r} \B_i$; with $\A=\A_\e \otimes \A_\r$ and $\B=\B_\e \otimes \B_\r$.

In each estimation round $i\in\N_\e$, the probability of obtaining $z_i$ is given by the expectation of the CHSH observable
\begin{align}\label{def:H_z}
  H_i(z_i) = \sum_{a_i,b_i,x_i,y_i}\frac 1 4 \,\delta^{z_i + 1}_{a_i + b_i + x_i y_i} A_i(a_i|x_i) B_i(b_i|y_i)\ ,
\end{align}
where $\delta$ denotes the Kronecker delta, producing the outcome $z_i=1$ when the CHSH game is won and $z_i=0$ when it is lost. Note that the algebra of the indices of $\delta$ is modulo 2, and that the operator $A_i(a_i|x_i)$ acts as the identity on all factors of the Hilbert space except $\A_i$, allowing us to write \eqref{def:H_z} without the tensor-product symbol (i.e., not writing $A_i(a_i|x_i) \otimes B_i(b_i|y_i)$). For more detail on the CHSH inequality, operator and observable, see \cite{seedless, nielsen2000quantum}.

For any set of estimation rounds $\mathcal{N}_\e$, the vector $\mathbf{z}_\e \in \{0,1\}^{n_\e}$ contains the results of the CHSH game on those rounds, and its Hamming weight $|\mathbf{z}_\e|$ gives the number of winning rounds.
This information determines one of the following events: 
\begin{align}
  &\bar{\mathfrak a} =\mbox{``protocol not aborted", if }|\z_\e|\geq n_\e w_\e\ ,
  \\
  &\mathfrak a =\mbox{``protocol aborted", if } |\z_\e| < n_\e w_\e\ , 
\end{align}
where $w_\e \in (w_\c, w_\q]$ is a parameter fixed at the start of the protocol.

The state of the key-generation rounds conditioned on not aborting the protocol is 
\begin{align}
  \rho_{\A_\r \B_\r \E|\bar{\mathfrak a}} 
  =\frac {\tr_\e(\rho_{\A\B\E} W_\e)} {\tr(\rho_{\A\B\E} W_\e)}\ ,  
\end{align}
where we use the notation $\tr_\e = \tr_{\A_\e \B_\e}$ and the no-abort POVM element
\begin{align}\label{def:W_e}
  W_\e =\sum_{\z_\e: |\z_\e|\geq n_\e w_\e}\  \prod_{i\in \N_\e} H_i(z_i)\ .
\end{align}
Instead of conditioning on $\mathfrak a$ or $\bar{\mathfrak a}$ we will use the subnormalised states 
\begin{align}
  \rho_{\A_\r \B_\r \E,\bar{\mathfrak a}} &= p(\bar{\mathfrak a}) \rho_{\A_\r \B_\r \E|\bar{\mathfrak a}} = \tr_\e(\rho_{\A\B\E} W_\e)\ ,
  \\
  \rho_{\A_\r \B_\r \E, \mathfrak a} &= p(\mathfrak a) \rho_{\A_\r \B_\r \E|\mathfrak a} = \tr_\e(\rho_{\A\B\E} [\unity-W_\e])\ ,
\end{align}
where $p(\bar{\mathfrak a}) = \tr(\rho_{\A\B\E} W_\e)$ and $p(\mathfrak a) = \tr(\rho_{\A\B\E} [\unity-W_\e])$ are the probabilities of not aborting and aborting respectively.

In each raw-key round $i\in\N_\r$, the system $\A_i$ is measured with $A_i(a_i|0)$ and the outcome $a_i$ is produced. After all measurements, the raw key vector $\a_\r \in \{0,1\}^{n_\r}$ containing all generation-round outcomes $a_i$ is multiplied by a full row-rank $m \times n_\r$ matrix $G$ (with entries in $\{0,1\}$) to produce the secret key $\k = G \a_\r \ (\bmod\ 2)$, where $m \leq n_\r$. Notably, this extraction step can be computed in just $O(mn_{\r})$ time.
For our main result, the matrix $G$ must be chosen so that no vector $\br \in \Span G$, where $\Span G \subseteq {0,1}^{n_\r}$ denotes the row space of $G$, has weight $|\br|$ such that
\begin{align}\label{extract condit}
    m < n_\r - \log_2 \binom{n_\r}{|\br|} - \log_2(2n_\r)\ ,
\end{align}
where $m$ is the final key length output by the protocol, computed using formula~\eqref{main result}.
The existence of such matrices is proven in Lemma~\ref{lem: 5}.

The completely positive trace-preserving (CPTP) map representing this process,
$\Upsilon_\r : \mathfrak{B}(\A_\r \B_\r) \to \mathfrak{B}(\K)$, where
$\mathfrak{B}(\mathcal{H})$ denotes the set of bounded linear operators on a Hilbert space $\mathcal{H}$, is defined as follows and is understood to act as the identity on any additional registers.
\begin{align}\label{eq:upsilon}
  \Upsilon_\r(\rho_{\A_\r \B_\r})
  = \sum_{\k, \a_\r} \proj{\k}_{\K}\, \delta_{G(\a_\r)}^{\k}\,
  \tr_{\r}\!\left[
    \rho_{\A_\r \B_\r}
    \prod_{i \in \N_\r} A_i(a_i | 0)
  \right].
\end{align}
where $\ket \k \in \K$ for all $\k\in\ZO^m$ is an orthonormal basis representing the classical register which contains the secret key.  
When the protocol is not aborted, the final global state is 
\begin{align}\label{eq:final no ab}
  \rho_{\K\E,\bar{\mathfrak a}|\N_\e} = \Upsilon_\r (\tr_\e (\rho_{\A\B\E} W_\e))\ .   
\end{align}
By conditioning on a fixed realisation of $\N_\e$, we can use the Hilbert-space factorisations $\A = \A_\e \otimes \A_\r$ and $\B = \B_\e \otimes \B_\r$, with fixed dimensions for $\a_\r$, $\z_\e$, and the matrix $G$.

\subsection{Security condition and theorem}

For our security proof, we want to consider the case in which Eve learns the estimation set $\N_\e$ after the protocol's round-execution step. In this scenario, the actual state produced by the protocol can be written as
\begin{align}\label{eq:actual}
  \rho^\mathsf{actual}_{\K\E\mathcal R,\bar{\mathfrak a}} = \EE_{\N_\e} \rho_{\K\E,\bar{\mathfrak a}|\N_\e}\, \proj{\N_\e}_{\mathcal R}\ , 
\end{align}
where we have taken the state in \eqref{eq:final no ab} and added a classical register $\mathcal R$ specifying the estimation subset in an orthonormal basis $\{\ket{\N_\e}\}$.
We will show that the protocol is secure even if Eve learns $\mathcal R$ after the implementation of the protocol. So Eve's final system is $\E\mathcal R$.

In contrast to the actual key \eqref{eq:actual}, an ideal (subnormalised) secret key state is
\begin{align}\label{def:ideal}
  \rho^\mathsf{ideal}_{\K\E\mathcal R,\bar{\mathfrak a}} = u_{\K} \EE_{\N_\e} \rho_{\E,\bar{\mathfrak a}|\N_\e}\, \proj{\N_\e}_{\mathcal R}\ ,
\end{align}
where
\begin{align}
  u_{\K} = 2^{-m} \sum_\k \proj{\k}_{\K}
\end{align}
is the uniform (maximally mixed) state on the classical register associated
with the final key $\mathcal{K}$, and $\rho_{\E,\bar{\mathfrak a}|\N_\e}$ is the $\K$-trace of~\eqref{eq:final no ab}.
In general, the states in \eqref{eq:actual} and \eqref{def:ideal} are not equal, but the standard security condition in quantum cryptography \cite{renner2005, portmann2022} is to impose that they are near-indistinguishable, that is,
\begin{align}\label{eq:sec cond}
  \left\|\rho^\mathsf{actual} _{\K\E\mathcal R,\bar{\mathfrak a}} -\rho^\mathsf{ideal} _{\K\E\mathcal R,\bar{\mathfrak a}} \right\|_1 
  \leq \epsilon\ ,
\end{align}
where the error parameter $\epsilon>0$ can be chosen by the user.
This security condition guarantees that using the actual key in any application will produce the same statistics as an ideal key, up to an event with probability $\epsilon$. This feature is called universally-composable security \cite{canetti2001universally}.

For small $\epsilon$, condition~\eqref {eq:sec cond} implies that the protocol either aborts with high probability or produces a key that is indistinguishable from uniform.
In reality, our protocol can be implemented with high no-abortion probability. For example, if we have an honest implementation that uses the same state $\rho_{\A_1\B_1}$ and the same measurements $A_1(a|x), B_1(b|y)$ in every round, and choose a value of $w_\e$ strictly smaller than $\tr(\rho_{\A_1\B_1} H_1(1))$, then the abort probability becomes exponentially small in $n$.

We now state the main result of this work. To do so, we use the modified
binary entropy function
\begin{align}
    \tilde h(q)
    :=
    \begin{cases}
        h(q),
        &
        0
    \le
    q< \frac{1}{np_\r}\Bigl\lfloor\frac{np_\r}{2}-\sqrt{\frac{n}{8}\ln\frac4\epsilon}\Bigr\rfloor
        \\[1.2ex]
        1,
        &
        q\geq 
        \frac{1}{np_\r}\Bigl\lfloor\frac{np_\r}{2}-\sqrt{\frac{n}{8}\ln\frac4\epsilon}\Bigr\rfloor
    \end{cases}\ ,
    \label{eq: modified_h}
\end{align}
where $h(q) = - q\log_2 q - (1-q) \log_2 (1-q)$ is the binary entropy function and $\lfloor q\rfloor = \max\{j\in\mathbb Z: j\leq q\}$ denotes the floor of a real number $q$. Note that, for any fixed $\epsilon$, $\tilde h(q)$ converges to $h(q)$ as $n\to\infty$ for all $0\leq q<1/2$.
\begin{theorem}
  For any choice of the parameters $n\in \mathbb Z^+$, $\epsilon \in (0,1]$ and $w_\e \in (w_\c + \Delta, w_\q]$, the protocol \ref{protocol_description} with key length
  \begin{align}\label{main result}
    m
    =
    \left\lfloor np_\r
    \left[
        1-
        \tilde h\!\left(\frac {2\sqrt 2 +4 -8w_\e + 8\Delta} {p_\r (\sqrt 2 -1)}
        \right)
    \right]
    -
    \sqrt{\frac n2\ln\frac4\epsilon}
    -
    \log_2\!\left(np_\r\right)-4\right\rfloor\ ,
\end{align}
  and parameter values $p_\r = 1-n^{-1/3}$ and $\Delta = n^{-1/3} \sqrt{2\ln(8 /\epsilon)}$, produces a no-abortion state $\rho^\mathsf{actual}_{\K\E\mathcal R,\bar{\mathfrak a}}$ which satisfies the security condition \eqref{eq:sec cond}.
\end{theorem}
If $m \leq 0$ or no valid choice of $w_\e$ exists, no key can be produced.
We also note that the validity of \eqref{eq:sec cond} is independent of the choice of Hilbert spaces $\A_i, \B_i, \E$, the initial state $\rho_{\A\B\E}$ and the measurements $A_i(a_i | x_i)$ and $B_i(b_i | y_i)$, although these choices can influence the probability of aborting.

For simplicity, we fix the parameters $\Delta$ and $p_{\r}$ to the values stated above. However, these choices are not unique, and further optimisation could substantially improve performance for small $n$, which we leave for future work. Even so, with these selections, a positive key length becomes possible for maximal CHSH violation whenever $n \geq 572{,}998$. 

The key length \eqref{main result} gives the asymptotic key rate
\begin{align}
  R(w_\e) = \lim_{n \to \infty} \frac{m}{n} \geq 1 -  h\!\left[ \frac{2\sqrt 2 + 4 - 8 w_\e}{\sqrt 2 - 1} \right]\ ,
\end{align}
for any constant $\epsilon$. This rate is shown in Figure~\ref{fig:rate}. 
A positive rate is obtained whenever the threshold winning frequency for the CHSH game exceeds $w_\e = \tfrac{9 + 3\sqrt 2}{16} \approx 0.828$ (corresponding to a CHSH value of roughly $2.621$), with a rate of $1$ at maximal CHSH violation.

\begin{figure}[h!]
  \centering
  \includegraphics [width=0.7\linewidth]{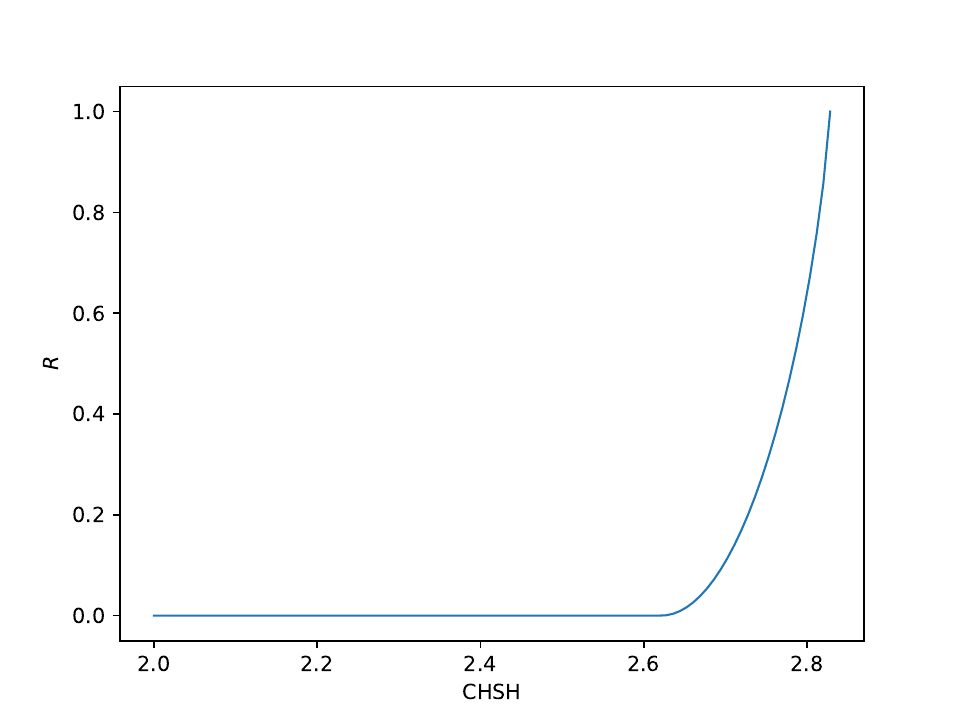}
  \caption{Key rate $R$ versus CHSH value ($8 w_\e-4$) in the large-$n$ regime.}
  \label{fig:rate}
\end{figure} 

\section{Proof of the theorem}
\label{sec:proof}
In this section, we prove our main theorem. First, we note that the block-diagonal structure of the operator
\begin{align} \label{eq: sec-op}
  &\rho^\mathsf{actual}_{\K\E\mathcal R,\bar{\mathfrak a}}
  -
  \rho^\mathsf{ideal}_{\K\E\mathcal R,\bar{\mathfrak a}} =
  \EE_{\N_\e}
  \left(
    \rho_{\K\E,\bar{\mathfrak a}|\N_\e}
    -
    u_{\K} \rho_{\E,\bar{\mathfrak a}|\N_\e}
  \right)\!
  \proj{\N_\e}_{\mathcal R}\ ,
\end{align}
implies 
\begin{align}
  \left\|
    \rho^\mathsf{actual}_{\K\E\mathcal R,\bar{\mathfrak a}}
    -
    \rho^\mathsf{ideal}_{\K\E\mathcal R,\bar{\mathfrak a}}
  \right\|_1
  =
  \EE_{\N_\e}
  \left\|
    \rho_{\K\E,\bar{\mathfrak a}|\N_\e}
    -
    u_{\K} \rho_{\E,\bar{\mathfrak a}|\N_\e}
  \right\|_1\ .
\end{align} 
In the first part of the proof, we analyse this trace norm for a fixed realisation of $\N_\e$, 
\begin{align}\label{eq:15}
  \|\rho_{\K\E,\bar{\mathfrak a}|\N_\e}-u_\K \rho_{\E,{\bar{\mathfrak a}}|\N_\e}\|_1
  =
  \left\|(\Upsilon_\r -u_\K \tr_\r)(\tr_\e(\rho_{\A\B\E} W_\e)) \right\|_1\ ,
\end{align}
and reintroduce the $\N_\e$ average later. 

We begin by showing that all quantum systems $\A_i$ and $\B_i$ can be assumed to be qubits without loss of generality. 

\subsection{Reduction to qubits}
In the DI framework, the adversary can pick the initial state $\rho_{\A\B\E}$ and all measurements $A_i(a_i|x_i)$ and $B_i(b_i|y_i)$, and, after Alice and Bob perform their measurements, holds the resulting classical-quantum state
\begin{align}\label{def:cq-state}
  \rho_{\E,\a,\b|\x,\y}
  = \tr_{\A\B}\!\left(\! \rho_{\A\B\E} \prod_{i=1}^n A_i(a_i|x_i) B_i(b_i|y_i) \right)\ .
\end{align}
Once $\N_\e$ is fixed, for convenience, we can reorder the Hilbert spaces so that the raw-key rounds are followed by the estimation rounds, allowing us to write
\begin{align}
  &\a =((a_i)_{i \in \mathcal{N}_\r}, (a_i)_{i \in \mathcal{N}_\e}) = (\a_\r, \a_\e)\ , 
  &&\b = ((b_i)_{i \in \mathcal{N}_\r}, (b_i)_{i \in \mathcal{N}_\e}) =(\b_\r, \b_\e)\ , \\ 
  &\x = ((0)_{i \in \mathcal{N}_\r}, (x_i)_{i \in \mathcal{N}_\e}) =(\mathbf 0, \x_\e)\ , 
  &&\y = ((0)_{i \in \mathcal{N}_\r}, (y_i)_{i \in \mathcal{N}_\e}) =(\mathbf 0, \y_\e)\ , \\
 &\z = ((0)_{i \in \mathcal{N}_\r}, (z_i)_{i \in \mathcal{N}_\e})=(\mathbf 0, \z_\e)\ , 
\end{align}
where $\mathbf{0} = (0, \ldots, 0) \in \{0, 1\}^{n_\r}$. Although the vector components have been reordered, the indices $i$, which denote the corresponding rounds, are preserved.

Once $w_\e$ and $G$ are fixed, the state $\rho_{\E,\a,\b|\x,\y}$ fully determines the value of the security parameter
\begin{align}\label{def:f}
  f\!\left[\rho_{\E,\a,\b|\x,\y}\right] &:= 
  \|\rho_{\K\E,\bar{\mathfrak a}|\N_\e}-u_\K \rho_{\E,{\bar{\mathfrak a}}|\N_\e}\|_1  
  \\ \nonumber &=
  \sum_{\k} \left\|\sum_{\a_\r,\b_\r}\!
  \left(\delta_{G\a_\r}^\k -2^{-m} \right)\!\!\! \sum_{\z_\e: |\z_\e|\geq n_\e w_\e}\, \sum_{\a_\e,\b_\e,\x_\e,\y_\e} \!\!\!\!\!\! 4^{-n_\e}\, \delta^{\z_\e +\mathbf 1}_{\a_\e +\b_\e +\x_\e \y_\e}\,
  \rho_{\E,\a,\b|\x,\y} \right\|_1,
\end{align}
where $\delta^{\z_\e +\mathbf 1}_{\a_\e +\b_\e +\x_\e \y_\e} = \prod_{i\in\N_\e} \delta^{z_i+1}_{a_i +b_i +x_i y_i}$.

\begin{lemma}\label{lemma:qubit redux}
  Any classical-quantum state $\rho_{\E,\a,\b|\x,\y}$ of the form \eqref{def:cq-state} with arbitrary Hilbert spaces $\A_i, \B_i, \E$ and binary inputs and outputs $a_i, x_i, b_i, y_i \in\ZO$ can be written as a mixture
  \begin{align}\label{eq:mixture}
    \rho_{\E,\a,\b|\x,\y} = \sum_\mathbf {v,w} p(\mathbf{v,w})\, \rho_{\E,\a,\b|\x,\y,\mathbf{v,w}} \ , 
  \end{align}
  with the following property. For every $\mathbf v=(v_1, \ldots, v_n)$ and $\mathbf w=(w_1, \ldots, w_n)$, the corresponding state
  \begin{align}\label{def:cq-state qubit}
    \rho_{\E,\a,\b|\x,\y,\mathbf{v,w}}
    = \tr_{\A'\B'}\!\left(\! \rho_{\A'\B'\E|\mathbf{v,w}} \prod_{i=1}^n A^{v_i}_i(a_i|x_i) B^{w_i}_i(b_i|y_i) \right)\ ,
  \end{align}
  is obtained by measuring qubits $\A' =\bigotimes_{i=1}^n \mathbb C^2$ and $\B' = \bigotimes_{i=1}^n \mathbb C^2$.
\end{lemma}

Using this lemma and applying the triangle inequality to the mixture \eqref{eq:mixture}, we obtain
\begin{align}
  f\!\left[ \rho_{\E,\a,\b|\x,\y} \right]  
  \leq \sum_\mathbf {v,w} p(\mathbf{v,w})\, f\!\left[ \rho_{\E,\a,\b|\x,\y,\mathbf{v,w}} \right]
  \leq \max_\mathbf {v,w} f\!\left[ \rho_{\E,\a,\b|\x,\y,\mathbf{v,w}} \right]\ ,
\end{align}
which implies that, without loss of generality, we can consider all of Alice and Bob's systems to be qubits (i.e., $\A_i \cong\B_i \cong\mathbb C^2$).
This allows us to prove the following.
\begin{lemma}\label{lem: 1}
  For $a,b,x,y \in \{0,1\}$ and any given qubit observables $A(a|x)$ and $B(b|y)$, the operator
  \begin{align}\label{def:H_z lemma}
    H(z)= \sum_{a,b,x,y}\frac 1 4 \,\delta^{z+1}_{a +b +x y}\, A(a|x) B(b|y)
  \end{align}
  has spectral decomposition
  \begin{align}\label{eq:H0 spectral}
    H(z) = \sum_{t,s=0}^1 \lambda_{t,s+z+1} \proj{\phi_{t,s}}\ ,
  \end{align}
  with the sum $s + z + 1$ taken modulo 2, and $\ket{\phi_{t,s}}$ are the eigenvectors associated to the eigenvalues
  \begin{align}\label{eq:lambda0010}
    &\lambda_{0,0} = \theta\ , 
    &\lambda_{1,0} = \frac 1 2 + \frac {\sqrt 2} 4 \sqrt{8\theta(1-\theta)-1}\ , 
    \\ \label{eq:lambda0111}
    &\lambda_{0,1} = 1-\theta\ ,
    &\lambda_{1,1} = \frac 1 2 - \frac {\sqrt 2} 4 \sqrt{8\theta(1-\theta)-1}\ , 
  \end{align}
  for some $\theta \in (w_\c , w_\q]$, satisfying $\lambda_{0,0} \geq\lambda_{1,0} \geq\lambda_{1,1} \geq\lambda_{0,1}$.
\end{lemma}

\subsection{State truncation}
Next, we want to truncate $\rho_{\A\B\E}$ onto a subspace of $\A\B$ that produces the outcome $W_\e$ with high probability, irrespective of $\N_\e$. 
The truncation is defined by the projector
\begin{align}\label{def:T}
  T= \sum_{(\t,\s)\in \mathcal T} \proj{\Phi_{\t,\s}}\ ,   
\end{align}
where $\mathcal T \subseteq \ZO^{2n}$, $\ket{\Phi_{\t,\s}} = \bigotimes_{i=1}^n \ket{\phi_{t_i,s_i}}$ and the states $\ket{\phi_{t_i,s_i}} \in \A_i \B_i$ are the eigenvectors of $H_i(1)$ with the form given in \eqref{eq:H0 spectral}.
For now, we only use that $T$ and $W_\e$ commute, as they are diagonal in the same eigen-basis, and defer the definition of the set $\mathcal T$ until it is used.

Let $\rho^\lo_{\A\B\E} = \rho_{\A\B\E} -T\rho_{\A\B\E}T$. Using the decomposition $\rho_{\A\B\E} =T\rho_{\A\B\E}T +\rho^\lo_{\A\B\E}$ in the fixed $\mathcal{N}_{\e}$ trace norm~\eqref{eq:15} and applying the triangle inequality gives
\begin{equation}
\begin{aligned}
  \left\| (\Upsilon_\r -u_\K \tr_\r)(\tr_\e(\rho_{\A\B\E} W_\e)) \right\|_1
  \leq
  &\left\| (\Upsilon_\r -u_\K \tr_\r)(\tr_\e(T\rho_{\A\B\E}T W_\e)) \right\|_1 \phantom{\ .}
  \\ \label{eq:8}
  &+\left\| \Upsilon_\r (\tr_\e(\rho^\lo_{\A\B\E} W_\e)) \right\|_1 \\ &+\left\| u_\K \tr_\r (\tr_\e(\rho^\lo_{\A\B\E} W_\e)) \right\|_1\ .
\end{aligned}
\end{equation}
Using the contractivity of the trace norm under CPTP maps \cite{nielsen2000quantum}, the cyclicity of the partial trace for operators acting solely on the traced subsystem, and the commutativity of $T$ and $W_\e$, we can bound the second term as
\begin{align}
  \left\| \Upsilon_\r (\tr_\e (\rho^\lo_{\A\B\E} W_\e)) \right\|_1
  &\leq 
  \left\| \tr_\e(\rho^\lo_{\A\B\E} W_\e) \right\|_1
  \\ &=
  \left\| \tr_\e \!\left( W_\e^{1/2}(\rho_{\A\B\E} -T \rho_{\A\B\E}T) W_\e^{1/2} \right) \right\|_1\\
  &\leq \left\| W_\e^{1/2}\rho_{\A\B\E}W_\e^{1/2} -TW_\e^{1/2} \rho_{\A\B\E} W_\e^{1/2}T \right\|_1\ .\label{eq: pre_gentle_measurement_application}
\end{align}

By a variant of the gentle measurement lemma \cite[Section 9.4]{wilde2016quantum}, for any positive semidefinite operator $\rho$ with $\tr\rho\leq 1$ and any operator $0\leq \Lambda\leq \unity$, we have
\begin{align}
    \|\rho-\sqrt{\Lambda}\rho\sqrt{\Lambda}\|_1
    \leq 2\sqrt{\tr\!\left((\unity-\Lambda)\rho\right)} \ .
\end{align} 
Applying the above inequality with $\rho = W_\e^{1/2}\rho_{\A\B\E}W_\e^{1/2}$ and $\Lambda = T$, noting $0 \leq W_\e \leq \unity$ and $\sqrt{T}=T$, we obtain the upper bound
\begin{align}
    \left\| \Upsilon_\r (\tr_\e (\rho^\lo_{\A\B\E} W_\e)) \right\|_1
  &\leq 2 \sqrt{ \tr((\mathds{1} - T) W_{\e}^{1/2} \rho_{\A \B \E} W_\e^{1/2})} \\
  &= 
  2\sqrt{\tr(\rho_{\A\B}(\unity-T)W_\e)}\ ,
  \label{eq:sqrt_bound_2nd_term}
\end{align}
where the final equality again uses that $T$ and $W_\e$ commute.

The third term in \eqref{eq:8} can be analysed in the same way as the second term, and is therefore also bounded by \eqref{eq:sqrt_bound_2nd_term}. Combining the bounds on the second and third terms gives the following upper bound on the fixed $\mathcal{N}_\e$ trace norm in \eqref{eq:15},
\begin{equation}
\begin{aligned}\label{eq:first second}
  \|\rho_{\K\E,\bar{\mathfrak a}|\N_\e}-u_\K \rho_{\E,{\bar{\mathfrak a}}|\N_\e}\|_1 
  \leq 
  &\left\| (\Upsilon_\r -u_\K \tr_\r)(\tr_\e(T\rho_{\A\B\E}T W_\e)) \right\|_1 \\ 
  &+ 4 \sqrt{\tr(\rho_{\A\B} (\unity-T) W_\e)}
\end{aligned}\ .
\end{equation}
The above bound has two contributions: the first corresponds to the security criterion under the constraint that the underlying unknown state used in the protocol has been truncated, while the second comes from the truncation error itself, i.e., that the truncated state is not too distinguishable from the original. The remainder of the proof is devoted to deriving simpler bounds for these two contributions, restoring the averaging over $\mathcal{N}_\e$ and identifying a truncation operator (through the correct selection of the set $\mathcal{T}$) that balances them so that the security condition~\eqref{eq:sec cond} is satisfied. 

\subsection{Bound of the second term}
To obtain a bound on the second term on the right-hand side of \eqref{eq:first second}, we use the structure of $W_\e$ and the spectral decomposition of $H_i(z_i)$ obtained in Lemma~\ref{lem: 1}. 
For any $\gamma\geq 0$, we can write
\begin{align}
  \tr(\rho_{\A\B} (\unity-T) W_\e) 
  &=  
  \sum_{(\t,\s) \notin \mathcal T}\! \bra{\Phi_{\t,\s}} \rho_{\A\B} \ket{\Phi_{\t,\s}} \sum_{\z_\e: |\z_\e| \geq n_\e w_\e} \prod_{i\in \mathcal{N}_\e} \lambda_{t_i, s_i+z_i+1}
  \\ &\leq
  \sum_{(\t,\s) \notin \mathcal T}\! \bra{\Phi_{\t,\s}} \rho_{\A\B} \ket{\Phi_{\t,\s}} \sum_{\z_\e} e^{\gamma(|\z_\e|-n_\e w_\e)} \prod_{i\in \N_\e} \lambda_{t_i, s_i+z_i+1}
  \\ &=
  \sum_{(\t,\s) \notin \mathcal T}\! \bra{\Phi_{\t,\s}} \rho_{\A\B} \ket{\Phi_{\t,\s}}
  \sum_{\z_\e} \prod_{i\in \N_\e} e^{\gamma (z_{i} - w_{\e})} \lambda_{t_{i}, s_{i} + z_{i} + 1}
  \\&=
  \sum_{(\t,\s) \notin \mathcal T}\! \bra{\Phi_{\t,\s}} \rho_{\A\B} \ket{\Phi_{\t,\s}} \prod_{i\in \N_\e} \left[ e^{\gamma(1-w_\e)} \lambda_{t_i,s_i} +e^{-\gamma w_\e} \lambda_{t_i,s_i+1} \right]\ ,
\end{align}
putting back the $\N_\e$ average and using the concavity of $\sqrt{q}$ on $q\geq 0$, it holds that
\begin{align}
     \EE_{\N_\e}\!\sqrt{\tr(\rho_{\A\B} (\unity-T) W_\e)}\leq \sqrt{\EE_{\N_\e}\tr(\rho_{\A\B} (\unity-T) W_\e)}\ .
\end{align}
Since
\begin{align}\nonumber
  \EE_{\N_\e}\! \tr(\rho_{\A\B} (\unity-T) W_\e) 
  &\leq\!\!   
  \sum_{(\t,\s) \notin \mathcal T}\!\! \bra{\Phi_{\t,\s}} \rho_{\A\B} \ket{\Phi_{\t,\s}} \prod_{i=1}^n \left[p_\r +p_\e e^{\gamma(1-w_\e)} \lambda_{t_i,s_i} +p_\e e^{-\gamma w_\e} \lambda_{t_i,s_i+1} \right]
  \\ &\leq   
  \max_{(\t,\s) \notin \mathcal T}  \prod_{i=1}^n \left[p_\r +p_\e e^{\gamma(1-w_\e)} \lambda_{t_i,s_i} +p_\e e^{-\gamma w_\e} \lambda_{t_i,s_i+1} \right]\ ,
\end{align}
and, using the inequality $1+x\leq e^{x}$, we obtain
\begin{align}\label{max exp}
  \EE_{\N_\e}\! \sqrt{\tr(\rho_{\A\B} (\unity-T) W_\e)} 
  &\leq   
    \max_{(\t,\s) \notin \mathcal T} \exp( \frac{p_\e}{2} \sum_{i=1}^n \left[ e^{\gamma(1-w_\e)} \lambda_{t_i,s_i} +e^{-\gamma w_\e} \lambda_{t_i,s_i+1} -1\right]).
\end{align}
Now, recall that in the DI framework, the set of eigenvalues $\lambda_{t_i,s_i}$ and eigenvectors $\ket{\phi_{t_i,s_i}}$ can vary in each round $i\in\{1,\ldots, n\}$.
To upper-bound \eqref{max exp}, we use the parametrisation of the eigenvalues from (\ref{eq:lambda0010}-\ref{eq:lambda0111}) and independently maximise over $\theta$ in each round. 
This maximisation becomes identical in each round and can be written as
\begin{align}
    \label{eq:xi}
  \xi_{t,s} 
  = \max_{\theta \in (w_\c , w_\q]} \left(e^{\gamma (1-w_\e)} \lambda_{t,s}(\theta) +e^{-\gamma w_\e}\lambda_{t,s+1}(\theta) \right)\ ,
\end{align} 
where $\lambda_{t,s}(\theta) = \lambda_{t,s}$ in (\ref{eq:lambda0010}-\ref{eq:lambda0111}). Putting $\xi_{t,s}$ in the upper bound \eqref{max exp} gives
\begin{align} \label{eq:45}
  \EE_{\N_\e}\! \sqrt{\tr(\rho_{\A\B} (\unity-T) W_\e)} 
  \leq 
  \max_{(\t,\s) \notin \mathcal T} \exp[\frac{n p_\e}{2} \left(\sum_{t,s=0}^1 \mu_{t,s}(\t,\s) \xi_{t,s} - 1\right)]\ ,
\end{align}
where $\mu_{t,s}(\t,\s)$ is the relative frequency of the pair $(t,s)\in \ZO^2$ in the vector $(\t,\s)\in \ZO^{2n}$. The following lemma provides a simple expression for the exponent for a suitable choice of $\gamma$.

\begin{lemma}\label{lem: 2}
For $\xi_{t,s}$ defined in \eqref{eq:xi} and any $(\t,\s)\in\{0,1\}^{2n}$, there exists a $\gamma$ such that
\begin{align}
\sum_{t,s=0}^1 \mu_{t,s}(\t,\s)\,\xi_{t,s} - 1
\le
\begin{cases}
-\Delta^2, & \text{if } \Delta \ge 0\\
0, & \text{otherwise}
\end{cases}\ ,
\end{align}
where 
\begin{align}
\Delta
= w_\e - \frac{1}{2}
- \frac{1}{4}\Bigl(\sqrt{2}\,\mu_{0,0}(\t,\s)
- \mu_{0,1}(\t,\s)
+ \mu_{1,0}(\t,\s)\Bigr)
\end{align}
for $w_\e\in (w_\c,w_\q]$. 
\end{lemma}
Using this result, we can select an appropriate set of strings $\mathcal T$ that define the projector $T$ in \eqref{def:T} to bound \eqref{eq:45}. Let
\begin{align}\label{def:cal T}
  \mathcal T_w = \Bigl\{ (\t,\s) : 
    w - \frac{1}{2} - \frac{1}{4} \Bigl(\sqrt{2}\,\mu_{0,0}(\t,\s) - \mu_{0,1}(\t,\s) + \mu_{1,0}(\t,\s)\Bigr) \leq 0
  \Bigr\}\ 
\end{align} 
for any $w\in (w_\c,w_\q]$. From now on, we make the dependence of $T$ on the parameter $w$ (through the set of vectors $\mathcal{T}_w$) explicit, denoting it by $T_w$. Using this notation, we set $w=w_\e - \Delta$ for any threshold winning frequency $w_\e \in (w_\c + \Delta, w_\q]$ and statistical uncertainty buffer $\Delta \in [0, w_{\e} - w_{\c})$, allowing us to write
\begin{align}\label{bound 2 term}
  \EE_{\N_\e} 4\sqrt{\tr\Bigl(\rho_{\A\B} (\unity - T_{w}) W_\e\Bigr)}
  \leq 4\exp\!\bigg[-\frac{n p_\e \Delta^2}{2}\bigg]\ .
\end{align} 

\subsection{Bound of the first term} 
\noindent 
Noting that $\sum_{\mathbf \a_\r} \prod_{i\in \N_\r} A_i(a_i|0)=\unity$, we can write
\begin{align}
  &\left\| (\Upsilon_\r -u_\K \tr_\r)(\tr_\e(T_w \rho_{\A\B\E}T_w W_\e)) \right\|_1 
  \\=& \label{eq: tracenorm_simplification}
  \sum_{\k} \left\|\sum_{\a_\r}
  \left(\delta_{G\a_\r}^\k -2^{-m} \right)
  \tr_{\A\B}\!\left( T_w \rho_{\A\B\E}T_w \left[\prod_{i\in \N_\r} A_i(a_i|0)\right] W_\e \right)
  \right\|_1
  \\ \label{eq:44} = &
  \sum_{\k} \left\|\sum_{\a_\r}
  \left(\delta_{G\a_\r}^\k -2^{-m} \right) 
  \sum_{\mathbf r} 2^{-n_\r} (-1)^{\a_\r\cdot\mathbf r}
  \tr_{\A\B}\!\left( T_w\rho_{\A\B\E}T_w \left[\prod_{i\in \N_\r} C_i^{r_i} \right] W_\e \right)
  \right\|_1 ,
\end{align}
where we have defined the Hermitian operators $C_i = A_i(0|0) -A_i(1|0)$. Recall that there is no loss of generality in assuming $A_i(a_i|0)$ are projectors (see Lemma~\ref{lemma:qubit redux}), so $C_i$ is full rank and satisfies $C_i^0=\unity$.

\begin{lemma}\label{lemma: indicator}
  Let $G$ be any full row-rank $m\times n_\r$ binary matrix and define the indicator function
  \begin{align} \label{def:indicator}
    I_G(\br) =
    \begin{cases}
    1, & \text{if } \br \in \Span G \text{ and } \br\neq 0\ , \\
    0, & \text{otherwise}
    \end{cases}
  \end{align}
  where $\Span G$ is the row space of $G$ (modulo 2). 
  The identity
  \begin{align} \label{eq:lemma-sum}
    2^{-n_\r} \sum_{\a_\r} \left( \delta_{G\a_\r}^{\k} -2^{-m}\right)(-1)^{\a_\r \cdot \br}
    = \pm 2^{-m} I_G(\br)
  \end{align}
  holds for any $\k \in \{0,1\}^m$ and $\br \in \{0,1\}^{n_\r}$.
\end{lemma}

Using the identity \eqref{eq:lemma-sum} in \eqref{eq:44} yields the upper bound
\begin{align}  
    &\left\| (\Upsilon_\r -u_\K \tr_\r)(\tr_\e(T_w \rho_{\A\B\E}T_w W_\e)) \right\|_1 \\
  \leq &\sum_{\br} I_G(\br) \left\|
  \tr_{\A\B}\! \left( \rho _{\A\B\E}\, T_w\! \left[\prod_{i\in \N_\r} C_i^{r_i} \right]\! T_w\, W_\e \right)\right\|_1 ,\label{eq:omega with E}
\end{align} 
where we have used cyclicity of the partial trace for operators acting only on $\A\B$, together with the commutativity of $T_w$ and $W_\e$, to rearrange the operator inside the trace.

Next, we show that the structure of $T_w$ imposes that $T_w [\prod_{i\in \N_\r} C_i^{r_i} ] T_w=0$ when $|\br|$ is sufficiently large. 

\begin{lemma}\label{lem: 3}
 For $a,b,x,y \in \{0,1\}$ and any given qubit observables $A(a|x)$ and $B(b|y)$, the eigenvectors $\ket{\phi_{t,s}}$ defined in \eqref{eq:H0 spectral} satisfy
  \begin{align}
    \left| \bra{\phi_{t,s}} C \ket{\phi_{t',s'}} \right| \leq \delta_{t+1}^{t'}\ ,
  \end{align}
  for all $t,t',s,s'\in \{0, 1\}$.
\end{lemma}

This lemma implies that any vector $\br$ satisfying 
\begin{align}\label{phi C^N phi}
  \bra{\Phi_{\t,\s}} \prod_{i\in \N_\r} C_i^{r_i} \ket{\Phi_{\t',\s'}} \neq 0\ ,    
\end{align}
for some $(\t,\s), (\t',\s') \in \mathcal T_{w}$ requires $|\br| =|\t+\t'| \leq |\t|+|\t'|$. Then, by the definition of $\mathcal T_{w}$, \begin{align}
  &4w-2 \leq \sqrt 2 \mu_{0,0}(\t,\s) -\mu_{0,1}(\t,\s) +\mu_{1,0}(\t,\s)\ ,
  \\  
  &\frac {|\t|} {n} 
  = \mu_{1,0}(\t,\s) +\mu_{1,1}(\t,\s)\ ,
\end{align}
with analogous expressions for $(\t',\s')$. This allows us to apply the following lemma.

\begin{lemma}\label{lem: 4}
If the relative frequencies $\mu_{t,s}(\mathbf{t},\mathbf{s})$ associated with $(\mathbf{t},\mathbf{s}) \in \{0,1\}^{2n}$ satisfy
\begin{align}
4w - 2 \leq
\sqrt{2}\,\mu_{0,0}(\mathbf{t},\mathbf{s})
- \mu_{0,1}(\mathbf{t},\mathbf{s})
+ \mu_{1,0}(\mathbf{t},\mathbf{s})\ ,
\end{align}
for some $w \in (w_\c, w_\q]$, then
\begin{align}
\mu_{1,0}(\mathbf{t},\mathbf{s}) + \mu_{1,1}(\mathbf{t},\mathbf{s})
\leq \tau(w)
:= \frac{\sqrt{2} + 2 - 4w}{\sqrt{2} - 1}\ .
\end{align}
\end{lemma}

Putting this all together, we have $|\t| = n(\mu_{1,0}(\mathbf{t},\mathbf{s}) + \mu_{1,1}(\mathbf{t},\mathbf{s}))$ and so $|\t| \leq n\tau(w)$ for all $(\t,\s) \in \mathcal T_{w}$. Together with the implication of Lemma~\ref{lem: 3}, we have $|\br| = |\t+\t'| \leq |\t| + |\t'| \leq 2n\tau(w)$, which allows us to replace the sum $\sum_{\br}$ in \eqref{eq:omega with E} by $\sum_{\br: |\br|\leq 2 n \tau(w)}$, i.e.,
\begin{align}  
    &\left\| (\Upsilon_\r -u_\K \tr_\r)(\tr_\e(T_w \rho_{\A\B\E}T_w W_\e)) \right\|_1 \\
  \leq &\sum_{k=1}^{\lfloor 2 n \tau(w)\rfloor} \sum_\br I_G(\br) \delta_{|\br|}^k \left\|
  \tr_{\A\B}\! \left( \rho _{\A\B\E}\, T_w\! \left[\prod_{i\in \N_\r} C_i^{r_i} \right]\! T_w\, W_\e \right)\right\|_1 .\label{eq: truncated_sum_over_r}
\end{align}

\begin{lemma}\label{lem: 5}
  For any pair of positive integers $(m, n_\r)$ with $m< n_\r$ and any $k \in \{1, \ldots, n_{\r}\}$, there exists a full row-rank $m\times n_\r$ matrix $G$ such that
  \begin{align}\label{eq: weight_vanish_condition2}
    \sum_{\br} I_G(\br)\,\delta_{|\br|}^{k} = 0
\end{align} 
whenever
\begin{align}\label{eq: weight_vanish_condition}
    m < n_\r - \log_2 \binom{n_\r}{k} - \log_2(2n_\r)\ ,
\end{align}
where the function $I_G(\br)$ is defined in \eqref{def:indicator}. 
\end{lemma}

We now choose $G$ so that no vectors in its row space have weight less than or equal to $\lfloor 2n\tau(w)\rfloor$. Lemma~\ref{lem: 5} shows that such a choice of $G$ is possible whenever $m$ lies below a threshold function, which we lower bound with
\begin{align}
    M(n_\r) = \begin{cases}
        n_\r
    -
    \log_2
    \binom{n_\r}{ \lfloor \min\{2n\tau(w),n_\r/2\}\rfloor}
    -
    \log_2(2n_\r), & n_\r \geq 1\\
    0, & \text{otherwise} 
    \end{cases}\ ,
    \label{eq: M_definition}
\end{align}
noting that
\begin{align}
    \binom{n_\r}{\lfloor\min\{2n\tau(w), n_\r/2\}\rfloor}
\end{align}
is the largest binomial coefficient
$\binom{n_\r}{k}$ for
$1 \leq k \leq \min\{\lfloor 2n\tau(w)\rfloor,n_\r\}$. We also note that, although Lemma~\ref{lem: 5} assumes that $m$ is a positive integer, we may allow any integer $m$, since the protocol produces no output when $m\leq 0$.

Returning to the trace norm, by Lemma~\ref{lem: 5}, 
$m<M(n_\r)$ gives
\begin{align}
    \sum_{\br} I_G(\br)\,\delta_{|\br|}^{k}=0,
    \qquad
    \forall\, k\in\{1,\ldots,\min\{2n\tau(w),n_\r\}\} \ , 
\end{align}
in which case
\begin{align}\label{eq: vanishing_tracenorm_observation}
    \left\|
    (\Upsilon_\r-u_\K\tr_\r)
    \bigl(
    \tr_\e(T_w\rho_{\A\B\E}T_w W_\e)
    \bigr)
    \right\|_1
    =0\ ,
\end{align}
and for arbitrary $m$ (i.e., including when $m \geq M(n_\r)$), the trace norm satisfies
\begin{align}\label{eq: general_tracenorm_bound}
    \left\|
    (\Upsilon_\r-u_\K\tr_\r)
    \bigl(
    \tr_\e(T_w\rho_{\A\B\E}T_w W_\e)
    \bigr)
    \right\|_1
    \le 2\ .
\end{align}

Now we put back the average over $\mathcal N_\e$ (equivalently $\mathcal N_\r$) in the trace norm and use these bounds. Let $N_\r\sim\mathrm{Bin}(n,p_\r)$ be the random variable corresponding to $n_\r$, so that
\begin{align}
    \Pr[N_\r=n_\r]
    =
    \binom{n}{n_\r}p_\r^{n_\r}p_\e^{n-n_\r}\ .
\end{align}
Then, we can write $\mathcal N_\e$-averaged trace norm as
\begin{align}
    &\EE_{\mathcal N_\e} \left\|
    (\Upsilon_\r-u_\K\tr_\r)
    \bigl(
    \tr_\e(T_w\rho_{\A\B\E}T_w W_\e)
    \bigr)
    \right\|_1 \\
    &=
    \sum_{n_\r=0}^n
    \Pr[N_\r=n_\r]\,
    \EE\!\left[
        \left\|
        (\Upsilon_\r-u_\K\tr_\r)
        \bigl(
        \tr_\e(T_w\rho_{\A\B\E}T_w W_\e)
        \bigr)
        \right\|_1
        \,\middle|\,
        N_\r=n_\r
    \right]
    \\
    &=
    \sum_{\substack{0\le n_\r\le n:\\ m<M(n_\r)}}
    \Pr[N_\r=n_\r]\,
    \EE\!\left[
        \left\|
        (\Upsilon_\r-u_\K\tr_\r)
        \bigl(
        \tr_\e(T_w\rho_{\A\B\E}T_w W_\e)
        \bigr)
        \right\|_1
        \,\middle|\,
        N_\r=n_\r
    \right]
    \nonumber\\
    &\quad+
    \sum_{\substack{0\le n_\r\le n:\\ m\ge M(n_\r)}}
    \Pr[N_\r=n_\r]\,
    \EE\!\left[
        \left\|
        (\Upsilon_\r-u_\K\tr_\r)
        \bigl(
        \tr_\e(T_w\rho_{\A\B\E}T_w W_\e)
        \bigr)
        \right\|_1
        \,\middle|\,
        N_\r=n_\r
    \right]
    \\
    &\le
    \sum_{\substack{0\le n_\r\le n:\\ m<M(n_\r)}}
    \Pr[N_\r=n_\r]\times 0
    +
    \sum_{\substack{0\le n_\r\le n:\\ m\ge M(n_\r)}}
    \Pr[N_\r=n_\r]\times 2
    \label{eq: zero_and_2}
    \\
    &=
    2
    \sum_{\substack{0\le n_\r\le n:\\ m\ge M(n_\r)}}
    \Pr[N_\r=n_\r]\ ,
    \label{eq: term1 upperbound}
\end{align}
where the expectation without subscript is over all sets $\mathcal{N}_\e$ of size $n_\e=n-n_\r$, conditioned on $n_\r$, and we have applied the bounds in \eqref{eq: vanishing_tracenorm_observation} and \eqref{eq: general_tracenorm_bound}.

Finally, we rewrite the probability \eqref{eq: term1 upperbound}. Since the value of $n_\r$ resulting from the protocol is not known in advance, we first define the output length $m$ in terms of a threshold $n_\r^*$, to be defined below. For given $n$, $w$, and $n_\r^*\in\{0,\ldots,n\}$, let
\begin{align}\label{eq: m_definition}
    m =\begin{cases}
        \lfloor M(n_\r^*)\rfloor-1, &\text{ if }2\lfloor 2n\tau(w)\rfloor+2\leq n_\r^* \leq n\\
     0, &\text{ otherwise}\ 
    \end{cases} \ . 
\end{align}
This choice ensures that $m<M(n_\r^*)$ whenever $2\lfloor 2n\tau(w)\rfloor+2\leq n_\r^* \leq n$, meaning we can use the following lemma. 
\begin{lemma}\label{lem: new_8}
    For any $n_\r^* \in \{0, \ldots, n\}$ such that $2\lfloor 2n\tau(w)\rfloor+2\leq n_\r^* \leq n$ and $m < M(n_\r^*)$, if $m \geq M(n_\r)$, then $n_\r\leq n_\r^*$.
\end{lemma}
Therefore, if $n_\r^*$ satisfies $2\lfloor 2n\tau(w)\rfloor+2\leq n_\r^* \leq n$, Lemma~\ref{lem: new_8} implies that $m\geq M(n_\r)$ only when $n_\r\leq n_\r^*$. This means that either the output length is 0, i.e., $m = 0$, or we can replace the summation condition in
\eqref{eq: term1 upperbound} to obtain the upper bound 
\begin{align}
    2\sum_{\substack{0\le n_\r\le n:\\ m\ge M(n_\r)}}
    \Pr[N_\r=n_\r]\
    \leq\ 2
    \sum_{\substack{0\le n_\r\le n:\\ n_\r\leq n_\r^*}}
    \Pr[N_\r=n_\r]\ =\ 2\Pr[N_\r\le n^*_\r]\ . \label{eq: term1 upperbound sum up to n_r*}
\end{align}

\subsection{Putting everything together}
To ensure the security condition \eqref{eq:sec cond}, we impose that, after reintroducing the $\mathcal{N}_\e$ average, each of the two terms in the upper bound on the trace norm in \eqref{eq:first second} is bounded by $\epsilon/2$. These two terms have been manipulated to obtain a bound on the second term, \eqref{bound 2 term}, and a bound on the first term, \eqref{eq: term1 upperbound sum up to n_r*}. To determine parameter choices that achieve this, we begin by fixing the value of the second term in the bound \eqref{bound 2 term} to
\begin{align}
  4\exp \bigg[-\frac{n p_\e \Delta^2}{2}\bigg] = \frac \epsilon 2\ .
\end{align}
This entails a trade-off between $p_\e$ and $\Delta$, for which we choose
\begin{align}
  &p_\e = n^{-1/3}\ , \qquad
  \Delta = n^{-1/3}\sqrt{2\ln(8 /\epsilon)}\ .
\end{align}
For the first term, we set $n_\r^*=\lfloor np_\r-\delta\rfloor$ and
\begin{align}\label{eq: delta_choice}
    \delta
    :=
    \sqrt{\frac n2\ln\frac4\epsilon} > 0\ .
\end{align}
With these parameter choices, we lower bound the value of $m$ defined in \eqref{eq: m_definition} to obtain the final key length,
\begin{align}\label{eq:finite_m_entropy_bound}
    \Bigg\lfloor np_\r
    \left[
        1-
        \tilde h\!\left(
            \frac{2\tau(w)}{p_\r}
        \right)
    \right]
    -
    \sqrt{\frac n2\ln\frac4\epsilon}
    -
    \log_2\!\left(np_\r\right)-4 \Bigg\rfloor \ ,
\end{align}
where $\tilde h(\cdot)$ is the modified
binary entropy defined by
\begin{align}
    \tilde h(q)
    :=
    \begin{cases}
        h(q),
        &
        0
    \le
    q< \frac{1}{np_\r}\Bigl\lfloor\frac{np_\r}{2}-\sqrt{\frac{n}{8}\ln\frac4\epsilon}\Bigr\rfloor
        \\[1.2ex]
        1,
        &
        q\geq 
        \frac{1}{np_\r}\Bigl\lfloor\frac{np_\r}{2}-\sqrt{\frac{n}{8}\ln\frac4\epsilon}\Bigr\rfloor
    \end{cases}\ .
    \label{eq: finite_modified_entropy}
\end{align} 
Appendix~\ref{choice of m proof} proves that $m$ is lower bounded by the final key length \eqref{eq:finite_m_entropy_bound}. This key length is positive only when $n_\r^*$ satisfies $2\lfloor 2n\tau(w)\rfloor+2\leq n_\r^*\leq n$.

Finally, we can use Hoeffding's inequality~\cite{Hoeffding1963} to bound the probability in \eqref{eq: term1 upperbound sum up to n_r*} as
\begin{align}
    2\Pr[N_\r\le n^*_\r]
    &=
    2\Pr[N_\r\le \lfloor np_\r-\delta\rfloor]
    \leq
    2\Pr[N_\r\le np_\r-\delta]
    \\
    &\le
    2\exp\left[-\frac{2\delta^2}{n}\right]
    \label{eq:hoeffding_bound_raw_count}
    \leq \frac \epsilon 2\ .
\end{align}
Under these choices, both terms in \eqref{eq:first second} are bounded by
$\epsilon/2$, and hence the total trace norm is at most $\epsilon$.

\section{Conclusion}
We have shown that randomness extraction in DI protocols can be implemented without a seed while remaining both efficiently implementable and high-rate. By introducing a truncation-based proof technique, we overcome the large statistical fluctuations that limited our previous work on seedless extractors \cite{seedless}, allowing the degree of Bell violation to be estimated using an asymptotically vanishing fraction of rounds. As a result, our protocol achieves the optimal rate of one key bit per singlet for the task of one-party DI randomness generation in the asymptotic limit, while using efficiently implementable linear functions as seedless extractors. This, however, comes with the trade-off that our analysis is tailored to the CHSH setting, arising from our use of Jordan's lemma–style techniques to reduce to qubits. Although our analysis can be directly extended to other Bell inequalities with binary inputs and outputs, such as the tilted CHSH inequalities in the bipartite scenario or the Mermin and Svetlichny inequalities in multipartite scenarios, extending our proof techniques to general Bell scenarios is nontrivial --- but we believe it should be possible, for instance by integrating our truncation approach at the level of correlation matrices in the Navascu\'es-Pironio-Ac\'in (NPA) hierarchy \cite{navascues2007bounding}. That said, since the CHSH inequality is central to many DI protocols, our results already apply to a wide range of tasks. It would be interesting to explore their potential in other settings \cite{Pironio_2009, Masanes2011, Pironio_2013, Vazirani_2014, pironio2010random, Colbeck_2011, Ac_n_2016, Colbeck_2012, brandao2016realistic, kessler2020device, foreman2023practical}, as well as whether they extend to non-DI scenarios with more trust \cite{pueyo2025deterministic}.

An important next step is to remove the assumption of memoryless measurement devices, extending our results to the fully general DI setting. Additionally, our result relies on a proof of the existence, rather than an explicit construction, of matrices $G$ satisfying~\eqref{extract condit}. This condition is equivalent to requiring that $G$ generate a linear error-correcting code with sufficiently large minimum distance, enabling explicit constructions based on such codes. A good candidate is the primitive narrow-sense Bose--Chaudhuri--Hocquenghem \cite{hocquenghem1959codes,bose1960class} codes, which admit quasilinear-time implementations~\cite{foreman2023thesis} but achieve a slightly worse key rate and do not support all values of $n_\r$.

\bigskip\noindent\textbf{Acknowledgements.} 
We acknowledge helpful discussions and comments from Erik Woodhead, Mafalda Almeida, Lewis Wooltorton, Tom\'{a}s Fern\'{a}ndez Martos and Gabriel Senno. SL acknowledges support from the EPSRC (grant number EP/S021582/1). LM acknowledges financial support from the EPSRC Prosperity Partnership in Quantum Software for Simulation and Modelling
(grant EP/S005021/1).

\bibliographystyle{unsrt}
\bibliography{library}

\appendix
\section{Proofs of the Lemmas}

\subsection{Proof of Lemma \ref{lemma:qubit redux}: Reduction to qubits}
\noindent\textbf{Lemma \ref{lemma:qubit redux}.}
Any classical-quantum state of the form
\begin{align}\label{def:cq-state ap}
  \rho_{\E,\a,\b|\x,\y}
  = \tr_{\A\B}\!\left(\! \rho_{\A\B\E} \prod_{i=1}^n A_i(a_i|x_i) B_i(b_i|y_i) \right)\ ,
\end{align}
with arbitrary Hilbert spaces $\A_i, \B_i, \E$ and binary inputs and outputs $a_i, x_i, b_i, y_i \in\ZO$ can be written as a mixture
  \begin{align}\label{eq:mixture A}
    \rho_{\E,\a,\b|\x,\y} = \sum_\mathbf {v,w} p(\mathbf{v,w})\, \rho_{\E,\a,\b|\x,\y,\mathbf{v,w}} \ , 
  \end{align}
  with the following property. For every $\mathbf v=(v_1, \ldots, v_n)$ and $\mathbf w=(w_1, \ldots, w_n)$, the corresponding state
  \begin{align}\label{def:cq-state qubit A}
    \rho_{\E,\a,\b|\x,\y,\mathbf{v,w}}
    = \tr_{\A'\B'}\!\left(\! \rho_{\A'\B'\E|\mathbf{v,w}} \prod_{i=1}^n A^{v_i}_i(a_i|x_i) B^{w_i}_i(b_i|y_i) \right)\ ,
  \end{align}
  is obtained by measuring qubits $\A' =\bigotimes_{i=1}^n \mathbb C^2$ and $\B' = \bigotimes_{i=1}^n \mathbb C^2$.

\begin{proof}
    By Naimark's Dilation Theorem \cite{naimark1943representation}, for each round $i$, we can introduce ancillary systems $\bar\A_i=\mathbb C^2$ and $\bar\B_i=\mathbb C^2$ with the state $\proj0_i$. The $n$-round system is then given by $\bar\A=\bigotimes_{i=1}^n\bar\A_i=(\mathbb C^2)^{\otimes n}$ and $\bar\B=\bigotimes_{i=1}^n\bar\B_i=(\mathbb C^2)^{\otimes n}$, with respective states $\sigma_{\bar\A} = \bigotimes_{i=1}^n \proj0_i$ and $\sigma_{\bar\B} = \bigotimes_{i=1}^n \proj0_i$, such that the POVMs $\bar\A=\bigotimes_{i=1}^n\bar\A_i=(\mathbb C^2)^{\otimes n}$ and $\bar\B=\bigotimes_{i=1}^n\bar\B_i=(\mathbb C^2)^{\otimes n}$ with respective states $\sigma_{\bar\A} = \bigotimes_{i=1}^n \proj0_i$ and $\sigma_{\bar\B} = \bigotimes_{i=1}^n \proj0_i$ such that the POVMs $A_i(a_i|x_i)$ and $B_i(b_i|y_i)$ are realised as projective measurements on the enlarged spaces
    \begin{align}
        \tilde\A_i=\A_i\otimes \bar\A_i\ , \quad \tilde\B_i=\B_i\otimes \bar\B_i\ .
    \end{align}
    Formally, let $\tilde\A=\bigotimes_{i=1}^n\tilde\A_i$ and $\tilde\B=\bigotimes_{i=1}^n\tilde\B_i$ and define 
    \begin{align}
        \tilde\rho_{\tilde\A\tilde\B\E}:=\rho_{\A\B\E}\otimes \sigma_{\bar\A}\otimes \sigma_{\bar\B}\ ,
    \end{align}
    then there exist projective measurements $\tilde A_i(a_i|x_i)$ on $\tilde \A_i$ and $\tilde B_i(b_i|y_i)$ on $\tilde \B_i$ that allow us to write the state \eqref{def:cq-state} as
    \begin{align}
        \rho_{\E,\a,\b|\x,\y} =\tr_{\tilde\A\tilde\B}\left(\tilde\rho_{\tilde\A\tilde\B\E}\prod_{i=1}^n \tilde A_i(a_i|x_i)\tilde B_i(b_i|y_i)\right)\ .
    \end{align}
    Following \cite{Masanes_2006}, for each $i$, by Jordan's lemma applied to Alice's two projectors $\tilde A_i(0|0)$ and $\tilde A_i(0|1)$ acting on $\A_i$, there exists a direct-sum decomposition $\tilde\A_i=\bigoplus_{v_i}\A_i^{v_i}$ where $\dim \A_i^{v_i}\leq 2$ such that both projectors are block diagonal with respect to this decomposition.
    Additionally, the identity $\tilde A_i(0|x_i) +\tilde A_i(1|x_i) =\unity$ implies that the same Jordan decomposition also block diagonalizes $\tilde A_i(1|x_i)$ for
    $x_i\in\{0,1\}$. This means that, by letting $P_i^{v_i}$ be the projector onto $\A_i^{v_i}$ and defining $A_i^{v_i}(a_i|x_i)
    :=
    P_i^{v_i}\tilde A_i(a_i|x_i)P_i^{v_i}$, we can write
    \begin{align}
    \tilde A_i(a_i|x_i)
    =
    \bigoplus_{v_i} A_i^{v_i}(a_i|x_i)\ .
    \end{align}
    Similarly, applying Jordan's lemma to Bob's projectors gives
    $\tilde{\B}_i
    =
    \bigoplus_{w_i}\B_i^{w_i}$, $
    \dim \B_i^{w_i}\leq 2$ 
    and $B_i^{w_i}(b_i|y_i)
    :=
    Q_i^{w_i}\tilde B_i(b_i|y_i)Q_i^{w_i}$, with $Q_i^{w_i}$ being the projector onto $\B_i^{w_i}$, so that
    \begin{align}
    \tilde B_i(b_i|y_i)
    =
    \bigoplus_{w_i} B_i^{w_i}(b_i|y_i)\ .
    \end{align}
    Next, let $\mathbf v=(v_1,\ldots,v_n)$, $\mathbf{w}=(w_1,\ldots,w_n)$ and define $P^{\mathbf v}:=\prod_{i=1}^n P_i^{v_i}$, $Q^{\mathbf w}:=\prod_{i=1}^n Q_i^{w_i}$. This allows us to write 
    \begin{align}
    \rho_{\E,\a,\b|\x,\y}
    &=
    \sum_{\mathbf v,\mathbf w}
    \tr_{\tilde{\A}\tilde{\B}}
    \left(
    P^{\mathbf v}Q^{\mathbf w}
    \tilde\rho_{\tilde{\A}\tilde{\B}\E}
    P^{\mathbf v}Q^{\mathbf w}
    \prod_{i=1}^n
    A_i^{v_i}(a_i|x_i)
    B_i^{w_i}(b_i|y_i)
    \right)\ 
    \end{align}
    and define the normalised state
    \begin{align}
    \rho_{\A'\B'\E|\mathbf v,\mathbf w}
    :=
    \frac{
    P^{\mathbf v}Q^{\mathbf w}
    \tilde\rho_{\tilde\A\tilde\B\E}
    P^{\mathbf v}Q^{\mathbf w}
    }{
    p(\mathbf v,\mathbf w)
    }\ ,
    \end{align}
    where $p(\mathbf v,\mathbf w) =\tr\left(P^{\mathbf v}Q^{\mathbf w}\tilde\rho_{\tilde\A\tilde\B\E}\right)$. This state is supported on $\A'=\bigotimes_{i=1}^n \A_i^{v_i}$ and $\B'=\bigotimes_{i=1}^n \B_i^{w_i}$. Since each $\A_i^{v_i}$ and $\B_i^{w_i}$ has dimension at most two, we may embed
    each block into a copy of $\mathbb C^2$, padding one-dimensional blocks
    arbitrarily. Thus each conditional state may be viewed as being generated by
    measurements on $\A'=(\mathbb C^2)^{\otimes n}$ and
    $\B'=(\mathbb C^2)^{\otimes n}$. Finally, by defining
    \begin{align}
    \rho_{\E,\a,\b|\x,\y,\mathbf v,\mathbf w}
    :=
    \tr_{\A'\B'}
    \left(
    \rho_{\A'\B'\E|\mathbf v,\mathbf w}
    \prod_{i=1}^n
    A_i^{v_i}(a_i|x_i)
    B_i^{w_i}(b_i|y_i)
    \right)\ , 
    \end{align} we arrive at
    \begin{align}
    \rho_{\E,\a,\b|\x,\y}
    =
    \sum_{\mathbf v,\mathbf w}
    p(\mathbf v,\mathbf w)
    \rho_{\E,\a,\b|\x,\y,\mathbf v,\mathbf w}\ .
    \end{align}
\end{proof}
We note that, in general, the direct sum decomposition is only possible when each party has just two projectors, which prevents the reduction to qubits in scenarios of more than two outcomes or more than two observables.

\subsection{Proof of Lemmas~\ref{lem: 1} and \ref{lem: 3}: Spectrum of Bell operators}
Recall that in every round $i$ Alice and Bob perform measurements $A_i(a_i|x_i)$ and $B_i(b_i|y_i)$ on the state $\rho_{\A\B\E}$ respectively, and we can consider these to be qubit observables using the results from the previous subsection. We remark that several of the results in this section are well known and can, for example, be inferred from \cite{Chefles_Barnett_1997,Braunstein_Mann_Revzen_1992}, but we include them here to make the presentation more self-contained.

The measurement statistics $p(\a,\b|\x,\y)$ are invariant under the application of local unitaries to Alice and Bob's systems. This means, without loss of generality, we can pick the local coordinate system for Alice (Bob) such that the measurement direction of the $x=0$ ($y=0$) setting aligns with the $x$-axis of the Bloch sphere, and use a single angle $\alpha \in (0, \pi]$ ($\beta \in (0, \pi]$) to specify the measurement direction of the $x=1$ ($y=1$) setting, relative to the $x$-axis. Formally, we define
\begin{align}\label{eq: alpha_beta A}
  \sigma_\alpha = \begin{pmatrix}
    0 & e^{i\alpha} \\
    e^{-i\alpha} & 0
  \end{pmatrix},
\end{align}
where $i$ is the imaginary unit. Define $A(a|0)$ for $a=0,1$ as the projectors onto the eigenvectors of $\sigma_0$, $A(a|1)$ those of $\sigma_\alpha$, $B(b|0)$ those of $\sigma_0$, and $B(b|1)$ those of $\sigma_\beta$, where $\alpha, \beta$ are free parameters chosen arbitrarily by the adversary and may vary for each round. With this parameterisation, we can write the CHSH operator as 
\begin{align}\label{eq: chsh_op_alpha_beta A}
    S(\alpha,\beta)=\sigma_0\otimes(\sigma_0+\sigma_\beta)+\sigma_\alpha\otimes(\sigma_0-\sigma_\beta)\ ,
\end{align}
By expanding \eqref{eq: chsh_op_alpha_beta A}, we can write 
\begin{align}
S(\alpha,\beta)&=\begin{pmatrix}
        & & & \mu\\
        & & \nu&\\
        & \bar\nu& &\\
        \bar\mu & & &
    \end{pmatrix}\ ,
\end{align}
where
\begin{align}
    \mu&=1+e^{i\beta}+(1-e^{i\beta})e^{i\alpha}\label{eq: mu}\ ,\\
    \nu&=1+e^{-i\beta}+(1-e^{-i\beta})e^{i\alpha}\label{eq: nu}\ ,
\end{align}
and $\bar\mu,\bar\nu$ denote their complex conjugates. It follows that the eigenvalues of $S(\alpha,\beta)$ are given by $\pm|\mu|,\pm|\nu|$, and, for $\alpha,\beta \in (0,\pi)$, the largest eigenvalue is
\begin{align}\label{eq: largest_CHSH_eigval A}
    |\mu|=2\sqrt{1+\sin\alpha\sin\beta}\ ,
\end{align}
second largest is 
\begin{align}\label{eq: 2nd_largest_CHSH_eigval A}
    |\nu|=2\sqrt{1-\sin\alpha\sin\beta}\ , 
\end{align}
with them equal whenever $\alpha = \pi$ or $\beta = \pi$.
Using (\ref{eq: mu}-\ref{eq: 2nd_largest_CHSH_eigval A}), we find that the eigenvectors (ordered by decreasing associated eigenvalue magnitude) are
\begin{align}
    \ket{\phi_{0,0}}&=\frac{1}{\sqrt 2}(\frac{\mu}{|\mu|}\ket{00}+\ket{11})\ ,\label{eq: eigenvectors_00}\\
    \ket{\phi_{1,0}}&=\frac{1}{\sqrt 2}(\frac{\nu}{|\nu|}\ket{01}+\ket{10})\ ,\\
    \ket{\phi_{1,1}}&=\frac{1}{\sqrt 2}(-\frac{\nu}{|\nu|}\ket{01}+\ket{10})\ ,\\
    \ket{\phi_{0,1}}&=\frac{1}{\sqrt 2}(-\frac{\mu}{|\mu|}\ket{00}+\ket{11})\label{eq: eigenvectors_11}\ .
\end{align}

\noindent\textbf{Lemma \ref{lem: 1}.} 
For $a,b,x,y \in \{0,1\}$ and any given qubit observables $A(a|x)$ and $B(b|y)$, the operator 
  \begin{align}\label{def:H_z lemma A}
    H(z)= \sum_{a,b,x,y}\frac 1 4 \,\delta^{z+1}_{a +b +x y}\, A(a|x) B(b|y)
  \end{align}
  has spectral decomposition
  \begin{align}\label{eq:H0 spectral A}
    H(z) = \sum_{t,s=0}^1 \lambda_{t,s+z+1} \proj{\phi_{t,s}}\ ,
  \end{align}
  with the sum $s + z + 1$ taken modulo 2, and $\ket{\phi_{t,s}}$ are the eigenvectors associated to the eigenvalues
  \begin{align}\label{eq:lambda0010 A}
    &\lambda_{0,0} = \theta\ , 
    &\lambda_{1,0} = \frac 1 2 + \frac {\sqrt 2} 4 \sqrt{8\theta(1-\theta)-1}\ , 
    \\ \label{eq:lambda0111 A}
    &\lambda_{0,1} = 1-\theta\ ,
    &\lambda_{1,1} = \frac 1 2 - \frac {\sqrt 2} 4 \sqrt{8\theta(1-\theta)-1}\ , 
  \end{align}
  for some $\theta \in (w_\c , w_\q]$, satisfying $\lambda_{0,0} \geq\lambda_{1,0} \geq\lambda_{1,1} \geq\lambda_{0,1}$.

\begin{proof}
Let $H_{z} := H(z)$. Using the parametrisation of $A(a|x)$ and $B(b|y)$ discussed below \eqref{eq: alpha_beta A}, we can write $H_z$ \eqref{def:H_z lemma A} in terms of $\alpha, \beta$. Decomposing the CHSH operator as
\begin{align}\label{eq: H0H1_relationship}
    S(\alpha,\beta) &= 4(H_1 - H_0) = 8H_1 - 4\unity = 4\unity - 8H_0\ ,
\end{align}
where the last two equalities follow from $H_1 + H_0 = \unity$, we see that $H_0$ and $H_1$ share the eigenvectors of $S(\alpha,\beta)$, namely $\ket{\phi_{0,0}}, \ket{\phi_{1,0}}, \ket{\phi_{1,1}}, \ket{\phi_{0,1}}$ 
from (\ref{eq: eigenvectors_00}–\ref{eq: eigenvectors_11}).

Let $\lambda_{0,0},\lambda_{1,0},\lambda_{1,1},\lambda_{0,1}$ denote the eigenvalues of $H_1$ corresponding to the eigenvectors $\ket{\phi_{0,0}},\ket{\phi_{1,0}},\ket{\phi_{1,1}},\ket{\phi_{0,1}}$. Using \eqref{eq: H0H1_relationship}, the eigenvalues of $H_1$ can be expressed in terms of those of $S(\alpha,\beta)$ as
\begin{align}\label{eq: H1_eigenvalues00}
    \lambda_{0,0}&=\frac{4+|\mu|}{8}=\frac{1}{4}(2+\sqrt{1+\sin\alpha\sin\beta})\ ,\\
    \lambda_{0,1}&=\frac{4-|\mu|}{8}=\frac{1}{4}(2-\sqrt{1+\sin\alpha\sin\beta})\ ,\\
    \lambda_{1,0}&=\frac{4+|\nu|}{8}=\frac{1}{4}(2+\sqrt{1-\sin\alpha\sin\beta})\ ,\\
    \lambda_{1,1}&=\frac{4-|\nu|}{8}=\frac{1}{4}(2-\sqrt{1-\sin\alpha\sin\beta})\ ,\label{eq: H1_eigenvalues11}
\end{align} 
and note that $\lambda_{0,0} + \lambda_{0,1} = \lambda_{1,0} + \lambda_{1,1} = 1$. Since $\sin\alpha \sin\beta \geq 0$ for all $\alpha, \beta \in (0, \pi]$, we have $\lambda_{0,0} \geq \lambda_{1,0} \geq \lambda_{1,1} \geq \lambda_{0,1}$. Setting $\lambda_{0,0} = \theta$ for some $\theta \in (w_\c, w_\q]$, the eigenvalues of $H_1$ (\ref{eq: H1_eigenvalues00}–\ref{eq: H1_eigenvalues11}) take the form (\ref{eq:lambda0010 A}–\ref{eq:lambda0111 A}).

The eigenvalues of $H_0$ corresponding to the eigenvector $\ket{\phi_{t,s}}$ for any $(t,s) \in \{0,1\}^2$ are $1 - \lambda_{t,s}$ (as $H_1 + H_0 = \unity$). This allows us to write the spectral decomposition of $H_z$ as
\begin{align}
    H_z = \sum_{t,s=0}^1 \lambda_{t,s+z+1} \proj{\phi_{t,s}}\ ,
\end{align}
with $s+z+1$ taken modulo 2, completing the proof.
\end{proof}

\noindent\textbf{Lemma \ref{lem: 3}.}
 For $a,b,x,y \in \{0,1\}$ and any given qubit observables $A(a|x)$ and $B(b|y)$, the eigenvectors $\ket{\phi_{t,s}}$ defined in \eqref{eq:H0 spectral A} satisfy
  \begin{align}
    \left| \bra{\phi_{t,s}} C \ket{\phi_{t',s'}} \right| \leq \delta_{t+1}^{t'}\ ,
  \end{align}
  for all $t,t',s,s'\in \{0, 1\}$.

\begin{proof}
Using the $\alpha,\beta$ parametrisation of Alice's projectors below \eqref{eq: alpha_beta A},
\begin{align}
    C=A(0|0)-A(1|0)=\sigma_0= \begin{pmatrix}
    0 & 1 \\
    1 & 0
\end{pmatrix}\ .
\end{align}
Let $\ket{\phi_{t,s}}$ be the eigenvectors from~(\ref{eq: eigenvectors_00}–\ref{eq: eigenvectors_11}).  
Since they are normalised and $\sigma_0$ has operator norm 1, we have 
$\bigl| \bra{\phi_{t,s}} \sigma_0 \ket{\phi_{t',s'}} \bigr| \leq 1$ for all $t,t',s,s' \in \{0,1\}$.  
By direct computation, for all $t,s,s' \in \{0,1\}$, we find $\bra{\phi_{t,s}} \sigma_0 \ket{\phi_{t,s'}} = 0$.
Combining these facts gives the desired bound,
\begin{align}
    \left| \bra{\phi_{t,s}} C \ket{\phi_{t',s'}} \right| \leq \delta_{t+1}^{t'}\ ,
  \end{align}
for all $t,t',s,s'\in \{0, 1\}$, noting that the algebra of the Kronecker delta indices is performed modulo 2.
\end{proof}

\subsection{Proof of Lemma \ref{lem: 2}}

\noindent\textbf{Lemma \ref{lem: 2}.}
For $\xi_{t,s}$ defined in \eqref{eq:xi} and any $(\t,\s)\in\{0,1\}^{2n}$, there exists a $\gamma$ such that
\begin{align}
\sum_{t,s=0}^1 \mu_{t,s}(\t,\s)\,\xi_{t,s} - 1
\le
\begin{cases}
-\Delta^2, & \text{if } \Delta \ge 0\\
0, & \text{otherwise}
\end{cases}\ ,
\end{align}
where 
\begin{align}
\Delta
= w_\e - \frac{1}{2}
- \frac{1}{4}\Bigl(\sqrt{2}\,\mu_{0,0}(\t,\s)
- \mu_{0,1}(\t,\s)
+ \mu_{1,0}(\t,\s)\Bigr)
\end{align}
for $w_\e\in (w_\c,w_\q]$. 

\begin{proof}
Recall that $\mu_{t,s}(\t,\s)$ is the relative frequency of the pair $(t,s)\in \ZO^2$ in the vector $(\t,\s)\in \ZO^{2n}$ and that $\xi_{t,s}  = \max_{\theta \in (w_\c , w_\q]} \left(e^{\gamma (1-w_\e)} \lambda_{t,s}(\theta) +e^{-\gamma w_\e}\lambda_{t,s+1}(\theta) \right)$
with $\lambda_{t,s}(\theta) = \lambda_{t,s}$ given in (\ref{eq:lambda0010}-\ref{eq:lambda0111}). The solutions of $\xi_{t,s}$ are 
\begin{align}
  \label{eq:xi00 xi10}
  &\xi_{0,0} =
  e^{-\gamma w_\e} \left[ (e^{\gamma}-1)w_\q +1 \right] \ , 
  &\xi_{1,0} = e^{-\gamma w_\e} \frac {3e^{\gamma}+1} 4\ ,
  \\ \label{eq:xi01 xi11}
  &\xi_{0,1} = e^{-\gamma w_\e} \frac {3+e^{\gamma}} 4\ ,
  &\xi_{1,1} = e^{-\gamma w_\e} \frac{1+e^{\gamma}} 2 \ .
\end{align}

Let $\xi=\sum_{t,s=0}^1 \mu_{t,s}(\t,\s) \xi_{t,s} -1$. Using the normalisation of $\{\mu_{t,s}(\t,\s)\}_{t,s}$, we can write 
\begin{align}\label{eq: min_beta}
  \xi
  = 
   \left[\frac 1 4 e^{\gamma(1-w_\e)}(2+\eta)
  +\frac 1 4 e^{-\gamma w_\e}(2-\eta) -1\right]\ ,
\end{align}
where $\eta = \sqrt 2 \mu_{0,0}(\t,\s) -\mu_{0,1}(\t,\s) +\mu_{1,0}(\t,\s)$. Differentiating \eqref{eq: min_beta} with respect to $\gamma$, we obtain
\begin{align}
    \frac{\partial\xi}{\partial\gamma}&=\frac{1-w_\e}{4}e^{\gamma(1-w_\e)}(2+\eta)+\frac{-w_\e}{4}e^{-\gamma w_\e}(2-\eta)\ .
\end{align}
Equating this to zero, we obtain the solution
\begin{align} \label{eq:beta0}
  \gamma_0 = \ln \frac {w_\e(2-\eta)} {(1-w_\e)(2+\eta)}\ .
\end{align}
Whenever $w_\e(2-\eta) \geq (1-w_\e)(2+\eta)$ (which is equivalent to $\eta \leq 4 w_\e - 2$), $\gamma_0$ is non-negative. We now consider two cases, $(i)$ when $\eta \leq 4 w_\e - 2$ and $(ii)$ when $\eta > 4 w_\e - 2$.

Case $(i)$: when $\eta \leq 4 w_\e - 2$. In this case, $\gamma_0$ minimises \eqref{eq: min_beta}. Note that the second derivative of \eqref{eq: min_beta} is always positive since $- 1 \leq \eta \leq \sqrt{2}$. Substituting \eqref{eq:beta0} back into $\xi$ then gives
\begin{align}\label{eq: beta_independent}  
  \xi=&\frac{2+\eta}{4}\left(\frac{w_\e}{1-w_\e}\frac{2-\eta}{2+\eta}\right)^{1-w_\e}+\frac{2-\eta}{4}\left(\frac{w_\e}{1-w_\e}\frac{2-\eta}{2+\eta}\right)^{-w_\e}-1
  \\=&
  \frac{4 w_\e -\Delta'}{4w_\e} \left(\frac  {(1-w_\e)(4w_\e -\Delta')}{w_\e(4-4w_\e+\Delta')}\right)^{w_\e-1} -1\ ,
\end{align}
where $\Delta'=4w_\e-2-\eta \geq 0$. This can be further simplified to
\begin{align}
\xi=&
  \frac{4 w_\e -\Delta'}{4w_\e} \left(\frac {1-w_\e}{w_\e}\frac{4w_\e}{4(1-w_\e)}\frac{1-\frac{\Delta'}{4w_\e}}{1+\frac{\Delta'}{4(1-w_\e)}}\right)^{w_\e-1} -1\\
  =&
  \frac{4 w_\e -\Delta'}{4w_\e} \left(\frac{1-\frac{\Delta'}{4w_\e}}{1+\frac{\Delta'}{4(1-w_\e)}}\right)^{w_\e-1} -1\\
  =&
  \left(1-\frac{\Delta'}{4w_\e}\right)^{w_\e}\left(1+\frac{\Delta'}{4(1-w_\e)}\right)^{1-w_\e} -1\\
\leq& \left(1-\frac {\Delta'} 4\right)\left(1+\frac {\Delta'} 4\right)-1=-\frac{\Delta'^2}{16}
\end{align}
where the last inequality follows from the fact that for $x\geq0$, the two inequalities $(1\mp x)^y\leq 1\mp xy$ hold whenever $0\leq y\leq 1$ and $\pm x\leq 1$. We require $\Delta'\leq 4w_\e$ for $\left(1-\frac{\Delta'}{4w_\e}\right)^{w_\e}$ to be real.
By defining $\Delta=\frac {\Delta'} 4$, we are done for the case of $\eta \leq 4w_\e-2$. \\

Case $(ii)$: when $\eta>4w_\e-2$. This case can be equivalently expressed as $\Delta<0$, and we can simply substitute $\gamma=0$ into \eqref{eq: min_beta} to arrive at $\xi=0$.
\end{proof}

\subsection{Proof of Lemma \ref{lem: 4}}
\noindent\textbf{Lemma \ref{lem: 4}.}
If the relative frequencies $\mu_{t,s}(\mathbf{t},\mathbf{s})$ associated with $(\mathbf{t},\mathbf{s}) \in \{0,1\}^{2n}$ satisfy
\begin{align}
4w - 2 \leq
\sqrt{2}\,\mu_{0,0}(\mathbf{t},\mathbf{s})
- \mu_{0,1}(\mathbf{t},\mathbf{s})
+ \mu_{1,0}(\mathbf{t},\mathbf{s})\ ,
\end{align}
for some $w \in (w_\c, w_\q]$, then
\begin{align}
\mu_{1,0}(\mathbf{t},\mathbf{s}) + \mu_{1,1}(\mathbf{t},\mathbf{s})
\leq \tau(w)
:= \frac{\sqrt{2} + 2 - 4w}{\sqrt{2} - 1}\ .
\end{align}

\begin{proof}
By definition, for any $(\t,\s) \in \{0,1\}^{2n}$ we have
\begin{align}
  \frac {|\t|} {n} 
  &= \mu_{1,0}(\t,\s) +\mu_{1,1}(\t,\s)\ ,\\
  \frac {|\s|} {n} 
  &= \mu_{0,1}(\t,\s) +\mu_{1,1}(\t,\s)\ .
\end{align}
Additionally, we have the normalisation condition $\sum_{t,s}\mu_{t,s}(\t,\s)=1$, 
from which we can write
\begin{align}
    &\frac {|\t|-|\s|} {n} 
  =-\mu_{0,1}(\t,\s)+\mu_{1,0}(\t,\s)\\
  &\frac {|\t|+|\s|-\t\cdot\s} {n} 
  = 1-\mu_{0,0}(\t,\s)\ ,
\end{align}
where $\t \cdot \s /n = \mu_{1,1}(\t, \s)$, noting that the dot product $\t \cdot \s$ is not taken modulo 2.
Substituting the above into $\sqrt 2 \mu_{0,0}(\t,\s) -\mu_{0,1}(\t,\s) +\mu_{1,0}(\t,\s) \geq 4w-2$, we arrive at
\begin{align}
    \sqrt{2}-\frac{(\sqrt{2}-1)|\t|+(\sqrt{2}+1)|\s|-\sqrt{2}\ \t\cdot\s}{n}&\geq 4w-2\\
    \to \label{eq:inequalityA51}\frac{(\sqrt{2}-1)|\t|+(\sqrt{2}+1)|\s|-\sqrt{2}\ \t\cdot\s}{n}&\leq \sqrt{2}-(4w-2)\ .
\end{align}
Since $(\sqrt{2}+1)|\s|-\sqrt{2}\ \t\cdot\s\geq 0$ for all $\t$ and $\s$ with equality when $\s = \mathbf{0}$, 
by manipulating the inequality \eqref{eq:inequalityA51}, we can conclude
\begin{align}
    \mu_{1,0}(\t,\s)+\mu_{1,1}(\t,\s) = \frac{|\t|}{n} 
    \leq \frac {\sqrt 2 -4w+2} {\sqrt 2 -1}\ .
\end{align}
\end{proof}

\subsection{Proof of Lemma~\ref{lemma: indicator}}

\noindent\textbf{Lemma \ref{lemma: indicator}.}
Let $G$ be any full row-rank $m\times n_\r$ binary matrix and define the indicator function
  \begin{align} \label{def:indicator A}
    I_G(\br) =
    \begin{cases}
    1, & \text{if } \br \in \Span G \text{ and } \br\neq 0 \\
    0, & \text{otherwise}
    \end{cases} ,
  \end{align}
  where $\Span G$ is the row space of $G$ (modulo 2). 
  The identity
  \begin{align} \label{eq:lemma-sum A}
    2^{-n_\r} \sum_{\a_\r} \left( \delta_{G\a_\r}^{\k} -2^{-m}\right)(-1)^{\a_\r \cdot \br}
    = \pm 2^{-m} I_G(\br)
  \end{align}
  holds for any $\k \in \{0,1\}^m$ and $\br \in \{0,1\}^{n_\r}$.

\begin{proof}
    Since $G$ has full row-rank, for every $\k\in\{0,1\}^m$ there exists $\a_\k\in \{0,1\}^{n_\r}$ such that
    $G\a_\k=\k$. Moreover, $\delta_{G\a_\r}^{\k}=1$ for any $\a_\r$ implies $\a_\r + \a_\k\in\ker G$, where we note that the vector addition is elementwise modulo 2.
    Hence, for any $\k$, 
    \begin{align}
        2^{-n_\r}\sum_{\a_\r}\delta_{G\a_\r}^{\k}(-1)^{\a_\r\cdot\br}
        &=2^{-n_\r}\sum_{\a_\r:\,\a_\r + \a_\k\in\ker G}(-1)^{\a_\r\cdot\br} \nonumber\\
        &=2^{-n_\r}(-1)^{\a_\k\cdot \mathbf{r}}\sum_{\bv\in\ker G}(-1)^{\bv\cdot\br}. \label{eq:character-sum}
    \end{align}
    By the rank--nullity theorem, $|\ker G| = 2^{n_\r - m}$. The sum $\sum_{\bv\in\ker G} (-1)^{\bv\cdot\br}$ from~\eqref{eq:character-sum} equals $2^{n_\r - m}$ if $\bv\cdot\br = 0 \!\!\mod 2$ for all $\bv \in \ker G$. Otherwise, it is $0$: if there exists $\bv_0 \in \ker G$ such that $\bv_0 \cdot \br \neq 0 \!\!\mod 2$, then the terms cancel in pairs, since for every $\bv \in \ker G$ the pair $(\bv, \bv + \bv_0)$ contributes $(-1)^{\bv\cdot\br} + (-1)^{(\bv + \bv_0)\cdot\br} = 0$. The condition $\bv\cdot\br = 0 \!\!\mod 2$ for all $\bv \in \ker G$ is equivalent to $\br \in \Span G$. Therefore,
    \begin{align}
    2^{-n_\r}\sum_{\a_\r}\delta_{G\a_\r}^{\k}(-1)^{\a_\r\cdot\br}
    =
    \begin{cases}
    2^{-m}(-1)^{\a_\k\cdot\br} & \text{if } \br\in\Span G,\\
    0 & \text{otherwise}.
    \end{cases}
    \end{align}
    Finally, for $\br\neq 0$,
    \begin{align}
    2^{-n_\r}\sum_{\a_\r}2^{-m}(-1)^{\a_\r\cdot\br}=0,
    \end{align}
    which gives
    \begin{align}
    2^{-n_\r}\sum_{\a_\r}\bigl(\delta_{G\a_\r}^{\k}-2^{-m}\bigr)(-1)^{\a_\r\cdot\br}
    =
    (-1)^{\a_\k\cdot\br}\,2^{-m} I_G(\br)=\pm 2^{-m} I_G(\br).
    \end{align}
\end{proof}

\subsection{Proof of Lemma~\ref{lem: 5}}
The sum in \eqref{eq: weight_vanish_condition2 A} counts the vectors of Hamming weight $k$ in $\Span G$. We show that the average value of this sum over uniformly random $m$-dimensional subspaces is less than one given $k$ satisfies \eqref{eq: weight_vanish_condition A}. Hence, at least one such subspace contains no vector of weight $k$.

\noindent\textbf{Lemma \ref{lem: 5}.}
For any pair of positive integers $(m, n_\r)$ with $m< n_\r$ and any $k \in \{1, \ldots, n_{\r}\}$, there exists a full row-rank $m\times n_\r$ matrix $G$ such that
  \begin{align}\label{eq: weight_vanish_condition2 A}
    \sum_{\br} I_G(\br)\,\delta_{|\br|}^{k} = 0
\end{align} 
whenever
\begin{align}\label{eq: weight_vanish_condition A}
    m < n_\r - \log_2 \binom{n_\r}{k} - \log_2(2n_\r)\ ,
\end{align}
where the function $I_G(\br)$ is defined in \eqref{def:indicator A}. 

\begin{proof}
Choose an $m$-dimensional subspace of $\{0,1\}^{n_\r}$ uniformly at random, and let the rows of $G$ be any basis of this subspace. Then $G$ has full row rank. Every $m$-dimensional subspace contains exactly $2^m-1$ nonzero vectors. By symmetry, every nonzero vector $\br\in\{0,1\}^{n_\r}$ therefore satisfies
\begin{align}
    \Pr\bigl[\br\in\Span G\bigr]
    =
    \frac{2^m-1}{2^{n_\r}-1}
    <
    2^{m-n_\r}\ ,
\end{align}
where the final inequality holds for all $m<n_\r$.
Since there are $\binom{n_\r}{k}$ vectors of weight $k$, by linearity of expectation over uniformly random $m$-dimensional subspaces,
\begin{align}
    \mathbb E\left[
        \sum_{\br}I_G(\br)\,\delta_{|\br|}^{k}
    \right]
    &=
    \binom{n_\r}{k}\frac{2^m-1}{2^{n_\r}-1} \\
    &<
    2^{m-n_\r}\binom{n_\r}{k}\ .
\end{align}
Condition \eqref{eq: weight_vanish_condition A} implies
\begin{align}
    2^{m-n_\r}\binom{n_\r}{k}
    <
    \frac{1}{2n_\r}
    <
    1\ .
\end{align}
Thus, the expected number of weight-$k$ vectors in $\Span G$ is less than one. Since this number is a nonnegative integer, it must equal zero for at least one choice of $G$. Hence, there exists a full row-rank $m\times n_\r$ matrix $G$ such that
\begin{align}
    \sum_{\br}I_G(\br)\,\delta_{|\br|}^{k}=0\ ,
\end{align}
as required.
\end{proof}

\subsection{Proof of Lemma~\ref{lem: new_8}}
\label{proof of lem 8}
\noindent\textbf{Lemma \ref{lem: new_8}.}
For any $n_\r^* \in \{0, \ldots, n\}$ such that $2\lfloor 2n\tau(w)\rfloor+2\leq n_\r^* \leq n$ and $m < M(n_\r^*)$, if $m \geq M(n_\r)$, then $n_\r\leq n_\r^*$

\begin{proof}
The proof follows by first showing that $M(n_\r)$ is monotone increasing in the regime $2\lfloor 2n\tau(w)\rfloor+2\leq n_\r \leq n$.
Recall that 
\begin{align}\label{eq: M_def A}
        M(n_\r) = \begin{cases}
            n_\r
        -
        \log_2
        \binom{n_\r}{ \lfloor \min\{2n\tau(w),n_\r/2\}\rfloor}
        -
        \log_2(2n_\r), & n_\r \geq 1\\
        0, & \text{otherwise} 
        \end{cases}\ .
    \end{align}
$M(n_\r)$ is monotone increasing if 
\begin{align}\label{eq: monoton_condition}
    M(n_\r+1)-M(n_\r)\geq 0\ .
\end{align}
First, we note that our range of interest for the lemma, $2\lfloor 2n\tau(w)\rfloor+2\leq n_\r$, is equivalent to $\lfloor 2n\tau(w)\rfloor<\lfloor n_\r/2\rfloor$. To see this, write $n_\r\geq 2\lfloor 2n\tau(w)\rfloor+2$ as $n_\r/2\geq \lfloor 2n\tau(w)\rfloor+1$, and since the right-hand side is an integer, this condition is equivalent to $\lfloor n_\r/2\rfloor \geq \lfloor 2n\tau(w)\rfloor+1$, i.e., $\lfloor n_\r/2\rfloor > \lfloor 2n\tau(w)\rfloor$. Moreover, since $2\lfloor 2n\tau(w)\rfloor+2\leq n_\r$ implies $n_\r \geq 1$, it follows from \eqref{eq: M_def A} that $M(n_\r)=n_\r-\log_2\binom{n_\r}{\lfloor 2n\tau(w)\rfloor}-\log_2 (2n_\r)$ and we can write the condition \eqref{eq: monoton_condition} as
\begin{align}
    &1+\log_2\left[\binom{n_\r}{\lfloor 2n\tau(w)\rfloor}\binom{n_\r+1}{\lfloor 2n\tau(w)\rfloor}^{-1}\right]+\log_2\frac{n_\r}{n_\r+1}\geq 0\\
    &\implies \log_2\frac{n_\r(n_\r+1-\lfloor 2n\tau(w)\rfloor)}{(n_\r+1)^2}\geq  -1\\
    &\implies n_\r^2-2\lfloor 2n\tau(w)\rfloor n_\r-1\geq 0\ , 
\end{align}
which solves to
\begin{align}
    n_\r&\geq \lfloor 2n\tau(w)\rfloor+\sqrt{\lfloor 2n\tau(w)\rfloor^2+1}\\
    &> 2\lfloor 2n\tau(w)\rfloor\ .
\end{align}
Since $n_\r$ is an integer, this is equivalent to the requirement $n_\r\geq 2\lfloor 2n\tau(w)\rfloor+1$.
The set of $n_\r$ satisfying $2\lfloor 2n\tau(w)\rfloor+2\leq n_\r \leq n$ satisfies $n_\r\geq 2\lfloor 2n\tau(w)\rfloor+1$, so $M(n_\r)$ is monotone increasing in the regime $2\lfloor 2n\tau(w)\rfloor+2\leq n_\r \leq n$.

Finally, assume that $M(n_\r)\leq m$ but $n_\r>n_\r^*$ for some $n_\r^*$, then by the monotonicity of $M(n_\r)$ in the range $2\lfloor 2n\tau(w)\rfloor+2\leq n_\r \leq n$, we have $M(n_\r)\geq M(n_\r^*)>m$, contradicting our assumption and completing the proof.
\end{proof}

\subsection{The final key length lower bounds $m$}\label{choice of m proof}
\noindent Recall that
\begin{align}
    M(n_\r) = \begin{cases}
        n_\r
    -
    \log_2
    \binom{n_\r}{ \lfloor \min\{2n\tau(w),n_\r/2\}\rfloor}
    -
    \log_2(2n_\r), & n_\r \geq 1\\
    0, & \text{otherwise} 
    \end{cases}\ ,
\end{align}
and, for any $n_\r^* \in \{0, \ldots, n\}$,
\begin{align}
    m =\begin{cases}
        \lfloor M(n_\r^*)\rfloor-1, &\text{ if }2\lfloor 2n\tau(w)\rfloor+2\leq n_\r^* \leq n\\
     0, &\text{ otherwise}\ 
    \end{cases} \ ,
\end{align}
where $\tilde h$ is our modified binary entropy function
\begin{align}
    \tilde h(q)
    :=
    \begin{cases}
        h(q),
        &
        0
    \le
    q< \frac{1}{np_\r}\Bigl\lfloor\frac{np_\r}{2}-\sqrt{\frac{n}{8}\ln\frac4\epsilon}\Bigr\rfloor
        \\[1.2ex]
        1,
        &
        q\geq 
        \frac{1}{np_\r}\Bigl\lfloor\frac{np_\r}{2}-\sqrt{\frac{n}{8}\ln\frac4\epsilon}\Bigr\rfloor
    \end{cases}\ .
    \label{eq: finite_modified_entropy A}
\end{align}
We claim that, for $n_\r^*=\lfloor np_\r-\sqrt{\frac n 2 \ln\frac 4 \epsilon}\rfloor$, 
the following inequality holds
\begin{align}\label{eq: finite_m A}
    \Bigg\lfloor np_\r
    \left[
        1-
        \tilde h\!\left(
            \frac{2\tau(w)}{p_\r}
        \right)
    \right]
    -
    \sqrt{\frac n2\ln\frac4\epsilon}
    -
    \log_2\!\left(np_\r\right)-4 \bigg\rfloor \leq m \ .
\end{align}

\begin{proof}
We prove this by considering two cases, case 1, where $n_\r^*
<
2\lfloor 2n\tau(w)\rfloor+2$ (i.e., when $m=0$), and case 2, where $n_\r^*
\ge
2\lfloor 2n\tau(w)\rfloor+2$.\\

\noindent \textbf{Case 1: $n_\r^* < 2\lfloor 2n\tau(w)\rfloor+2$.} Since $\left\lfloor \frac{n_\r^*}{2}\right\rfloor$ is an integer, $2n\tau(w)<\left\lfloor \frac{n_\r^*}{2}\right\rfloor$ is equivalent to $\lfloor 2n\tau(w)\rfloor+1\leq \left\lfloor \frac{n_\r^*}{2}\right\rfloor$, which is equivalent to $2\lfloor 2n\tau(w)\rfloor+2\leq n_\r^*$ (since $\lfloor 2n\tau(w)\rfloor$ is also an integer), so we have
\begin{align}\label{eq: equivalent_conditions}
    \frac{2\tau(w)}{p_\r}<\frac{1}{np_\r}\left\lfloor \frac{n_\r^*}{2}\right\rfloor\iff n_\r^*
    \geq 
    2\lfloor 2n\tau(w)\rfloor+2\ .
\end{align} 
This implies that, if $n_\r^*<2\lfloor 2n\tau(w)\rfloor+2$, then $\tilde{h}\left(\frac{2\tau(w)}{p_\r}\right)=1$, and the left-hand side of \eqref{eq: finite_m A} becomes
    \begin{align}
    -
    \sqrt{\frac n2\ln\frac4\epsilon}
    -
    \log_2\!\left(np_\r\right)-4 \leq 0\ ,
\end{align}
and since $m=0$, the inequality holds.\\

\noindent\textbf{Case 2: $n_\r^* \geq 2\lfloor 2n\tau(w)\rfloor+2$.} 
In this case, 
    \begin{align}
    \label{eq: m start}
        m
        \geq n_\r^*-\log_2\binom{n_\r^*}{\lfloor 2n\tau(w)\rfloor}-\log_2\left(2n_\r^* \right)-2\ ,
    \end{align}
    and we lower bound the three $n_\r^*$-dependent terms separately.
    \begin{enumerate}
        \item First term: 
        \begin{align}
            n_\r^* = \left\lfloor np_\r-\sqrt{\frac n 2 \ln\frac 4 \epsilon}\right\rfloor \geq np_\r-\sqrt{\frac n 2 \ln\frac 4 \epsilon}-1\ .
        \end{align}
        \item Second term:
        We can write
        \begin{align}
            \log_2\binom{n_\r^*}{\lfloor 2n\tau(w)\rfloor} &=
            \log_2\binom{\lfloor np_\r-\sqrt{\frac n 2 \ln\frac 4 \epsilon}\rfloor}{\lfloor 2n\tau(w)\rfloor} \\
            &\leq \left\lfloor np_\r-\sqrt{\frac n 2 \ln\frac 4 \epsilon}\right\rfloor h\left(\frac{\lfloor 2n\tau(w)\rfloor}{\lfloor np_\r-\sqrt{\frac n 2 \ln\frac 4 \epsilon}\rfloor}\right)\ ,
        \end{align}
        where $h$ is the binary entropy function. This tail bound holds since we are in case 2, i.e., $\lfloor \frac{\lfloor np_\r-\sqrt{\frac n 2 \ln\frac 4 \epsilon}\rfloor}{2}\rfloor >\lfloor 2n\tau(w)\rfloor$, which implies $\frac{\lfloor 2n\tau(w)\rfloor}{\lfloor np_\r-\sqrt{\frac n 2 \ln\frac 4 \epsilon}\rfloor} < 1/2$. Using the fact that $qh(\frac{q'}{q})$ is non-decreasing in $q$ for a fixed $q'$,
        \begin{align}
            \left\lfloor np_\r-\sqrt{\frac n 2 \ln\frac 4 \epsilon}\right\rfloor h\left(\frac{\lfloor 2n\tau(w)\rfloor}{\lfloor np_\r-\sqrt{\frac n 2 \ln\frac 4 \epsilon}\rfloor}\right) 
            &\leq np_\r h\left(\frac{\lfloor 2n\tau(w)\rfloor}{np_\r}\right) \\
            &\leq np_\r \tilde{h}\left(\frac{\lfloor 2n\tau(w)\rfloor}{np_\r}\right)\\
            &\leq np_\r \tilde{h}\left(\frac{2\tau(w)}{p_\r}\right)\ ,
        \end{align}
        where the second to last inequality uses the fact that $h(q)\leq \tilde{h}(q)$ for all $0\leq q\leq 1$ and the final inequality follows from the fact that $\tilde h(q)$ is non-decreasing in $q$. Therefore, we have
        \begin{align}
            -\log_2\binom{n_\r^*}{\lfloor 2n\tau(w)\rfloor}\geq -np_\r \tilde{h}\left(\frac{2\tau(w)}{p_\r}\right)\ .
        \end{align}
        \item Third term:
        \begin{align}
            -\log_2\left(2n_\r^*\right) &= -\log_2\left(2\left\lfloor np_\r-\sqrt{\frac n 2 \ln\frac 4 \epsilon}\right\rfloor\right)\\ 
            &\geq -\log_22np_\r\\
            &=-\log_2np_\r-1\ .
        \end{align}
    \end{enumerate}
    Replacing the $n_\r^*$-dependent terms on the right-hand side of \eqref{eq: m start} with the three bounds above and taking the floor completes the proof in this case.
\end{proof}

\end{document}